\documentclass[journal,twocolumn]{IEEEtran}
\usepackage{stfloats,color}

\usepackage{cite}

\usepackage[numbers]{natbib}

\usepackage{amsthm}

\usepackage{floatflt}

\usepackage{graphicx}
\usepackage{amsmath}
\usepackage{amssymb}

\usepackage{verbatim}

\usepackage{psfrag}

\usepackage{subfigure}
\usepackage{stfloats}
\usepackage{amsmath}
\usepackage{epstopdf}
\usepackage{algorithmic}

\newtheorem{theorem}{Theorem}

\newtheorem{lemma}[theorem]{Lemma}

\graphicspath{{./figures/}}

\begin{document}
\title{Energy-Aware Cooperative Wireless Networks with Multiple Cognitive Users}

\author{Mahmoud Ashour,~\IEEEmembership{Member,~IEEE,}~
M.Majid~Butt,~\IEEEmembership{Senior~Member,~IEEE,} \\
Amr Mohamed,~\IEEEmembership{Senior~Member,~IEEE,}
Tamer ElBatt,~\IEEEmembership{Senior~Member,~IEEE,}
Marwan Krunz,~\IEEEmembership{Fellow,~IEEE}

%
\thanks{The material in this paper has been presented in part in ISIT 2014, Honolulu, Hi, USA \cite{Ashour:ISIT}.}
\thanks{Mahmoud Ashour is with the Department of Electrical Engineering and Computer Science, The Pennsylvania State University, PA, USA. E-mail: mma240@psu.edu.}
\thanks{M. Majid Butt is with the Center for Future Networks, Trinity College, University of Dublin, Ireland. 
	E-mail: majid.butt@ieee.org.}
\thanks{Amr Mohamed is with Computer Science and Engineering Dept., Qatar University, Doha, Qatar. E-mail: amrm@qu.edu.qa.}
\thanks{Tamer Elbatt is with the Wireless Intelligent Networks Center, Nile University, Giza, Egypt and EECE Dept., Faculty of Engineering, Cairo University, Giza, Egypt. Email: telbatt@ieee.org.}
\thanks{Marwan Krunz is with the Department of Electrical and Computer Engineering, University of Arizona, AZ, USA. Email: krunz@ece.arizona.edu.}

}

\maketitle

\begin{abstract}
In this paper, we study and analyze cooperative cognitive radio networks 
with arbitrary number of secondary users (SUs). Each SU is considered a prospective relay for the primary user (PU) besides having its own data transmission demand. We consider a multi-packet transmission framework which allows multiple SUs to transmit simultaneously thanks to dirty-paper coding. 
We propose power allocation and scheduling policies that optimize the throughput for both PU and SU with minimum energy expenditure.
The performance of the system is evaluated in terms of throughput and delay under different opportunistic relay selection policies. Towards this objective, we present a mathematical framework for deriving stability conditions for all queues in the system. Consequently, the throughput of both primary and secondary links is quantified. Furthermore, a moment generating function (MGF) approach is employed to derive a closed-form expression for the average delay encountered by the PU packets.
Results reveal that we achieve better performance in terms of throughput and delay at lower energy cost as compared to equal power allocation schemes proposed earlier in literature. Extensive simulations are conducted to validate our theoretical findings.
\end{abstract}


\section{Introduction}
Cognitive radio networks have emerged as an efficient solution to the problem of spectrum scarcity and its under-utilization. In a cognitive radio network, the secondary users (SUs) exploit primary users' (PUs) period of inactivity to enhance their performance provided that PUs' performance remains unaffected. Depending on the mode of interaction of the primary and the secondary users, the cognitive radio networks are classified as interweave, underlay and overlay networks. In the last decade or so, the industry and academia has shown overwhelming interest in the application of cognitive radios in different networking solutions. Reference \cite{haykin2005cognitive} provides a comprehensive overview of the cognitive radio fundamentals and research activities.

On the other hand, cooperative diversity has been widely investigated in pursuit of combating multipath fading \cite{Tse,Kramer}. Incorporating cooperation into cognitive radio networks results in substantial performance gains in terms of throughput and delay for both primary and
secondary nodes \cite{Ephremedis}. The SUs help the PUs to transmit their data, and create opportunities for their own data transmission at the same time. The cooperation between the PUs and the SUs vary from just sharing information about queue states, channel state information (CSI), and primary packet transmission activity to the use of SUs as cognitive relays.
Typically, relaying is carried out over orthogonal channels due to the half-duplex communication constraint at the relays \cite{Tse}. However, some of the recent solutions overcome this limitation by accommodating simultaneous transmissions in a single slot \cite{Gamal,Krikidis,Krikidis2}. This is achieved through space-time coding \cite{Gamal} or dirty-paper coding (DPC) \cite{Krikidis,Krikidis2}. Conventionally, zero forcing and more recently prior zero forcing \cite{Song:2013} has been employed to mitigate the SU signal interference with the PU signals. On the other side, for cooperative cognitive radio networks with multiple SUs with their own data transmission demands, employing DPC allows one SU to transmit new data while the other SU helps the PU by relaying its data. Thus, the spectral efficiency of the system is enhanced.

In literature, there is a rich volume of recent work focusing on cooperation in cognitive relay networks. The benefits of cooperative relaying has been discussed and analyzed in \cite{Simeone:2007,Papadimitriou14,Han:TCOM2012}. In \cite{Simeone:2007}, authors derive the maximum sustained throughput of a single SU to maintain a fixed throughput for PU with and without relaying. They used a dominant system approach to guarantee the queue stability of both SU and PU while overcoming the queues interaction. A cognitive system comprising a single PU and multiple SUs along with multiple relays is considered in \cite{Han:TCOM2012}, where a proportion of the secondary relays help the PU in communication while a relay selection is performed from the remaining relays to give simultaneous access to the SU. The authors show that there exists an optimal number of cooperating relays with the PU that achieve optimal outage performance. In \cite{Jing:2014}, the authors also discuss a cognitive relay selection problem using optimal stopping theory. Reference \cite{Neely:2012} addresses a cognitive radio cooperation model where the SU can transmit its data along with primary transmission, but cooperates by deferring its transmission when the PU is transmitting. The authors in \cite{El-sherif:2011} address a cooperative cognitive relay network where both primary and secondary nodes use cognitive relays for data transmission. The relays help the PUs empty their queues fast and thereby, the throughput for the SUs increases as a result. SU throughput stability regions for cooperative cognitive networks have been derived for cooperative cognitive radio networks in different settings in \cite{Simeone:2007,Krikidis:2010}.

Krikidis \emph{et al.} address different protocols for a cognitive cooperative network and the stable throughput for both primary and the secondary networks is derived. In this paper, we adopt the model presented in \cite{Krikidis} and employ DPC. We consider a cognitive network with arbitrary number of SUs co-existing with a PU and sharing one common relay queue. We propose power allocation and scheduling policies that enhance the throughput of both primary and secondary links using the least possible energy expenditure. The summary of the main contributions of this work is as follows.
\begin{itemize}
  \item We propose an energy-efficient adaptive power (AP) allocation scheme for the SUs that enhances the throughput of both primary and secondary links. Energy-efficient transmission is achieved via exploiting instantaneous CSI to adapt the transmission powers at all SUs.
  \item We introduce two SU scheduling policies, which prioritize primary or secondary throughput enhancement according to the network requirements. We analyze the performance of both policies in conjunction with equal and adaptive power allocation schemes.     
  \item We develop a generic mathematical framework to derive closed-form expressions for both PU and SU throughput, and PU average delay. The mathematical analysis is performed for an arbitrary number of SUs coexisting with a PU. A detailed analysis is performed for each combination of power allocation and SU scheduling policies. 
We validate our theoretical findings via simulations. 
Results reveal that AP-based schemes yield superior performance compared to EP allocation proposed in \cite{Krikidis}, with significantly less energy cost. 
\end{itemize}

The rest of this paper is organized as follows. Section \ref{background} presents the information-theoretic background and preliminaries needed in
the sequel. Section \ref{system_model} introduces the system model and the proposed cooperation strategy. The opportunistic relay selection and power allocation strategies are presented in Section \ref{CS_PA} along with their mathematical analysis in Section \ref{sect:analysis}.
Numerical results are then presented in Section \ref{results}. Finally, concluding remarks are drawn in Section \ref{sect:conclusion}.

\section{Background and Preliminaries}\label{background}
\subsection{Dirty-paper coding}
\label{dpc}
DPC was first introduced in \cite{DPC} and we briefly state its implication. Consider a channel with output $\mathbf{y}=\mathbf{x}+\mathbf{q}+\mathbf{z}$, where $\mathbf{x}$, $\mathbf{q}$ and $\mathbf{z}$ denote the input, interference, and noise, respectively. The input $\mathbf{x} \in \mathbb{C}^{m}$ satisfies the power constraint $(1/m)\sum_{i=1}^{m}|x_{i}|^{2} \leq
P_{0}$. We assume that $\mathbf{q}$ and $\mathbf{z}$ are zero-mean Gaussian vectors with covariance matrices $Q\mathbf{I}_{m}$ and
$N_{0}\mathbf{I}_{m}$, respectively, where $\mathbf{I}_{m}$ denotes the $m \times m$ identity matrix. If the interference $\mathbf{q}$ is unknown to both transmitter and receiver, the channel capacity is given by $\log(1+P_{0}/(Q+N_{0}))$ (bits/channel use).
However, if $\mathbf{q}$ is known to the transmitter but not the receiver, the channel capacity is shown to be the same as that of a standard "interference free" Gaussian channel with signal-to-noise ratio $P_{0}/N_{0}$ using DPC.
In other words, if the interference is known a priori at the transmitter, DPC renders the link between the transmitter and its intended receiver interference-free.

\subsection{Channel outage}\label{Channel_Outage}
We present the basic definition of an outage event and the corresponding outage probability calculation. Consider a channel with output
$\mathbf{y}=\sqrt{\mathbf{h}}\mathbf{x}+\mathbf{z}$, where $\sqrt{\mathbf{h}}$ and $\mathbf{x}$ denote the fading coefficient and the input, respectively. Moreover, the noise $\mathbf{z}$ is modelled as zero-mean circularly symmetric complex Gaussian random variable with variance $N_{0}$. For a target transmission rate $R_{0}$, an outage occurs if the mutual information between the input and output is not sufficient to support that rate. The probability of such event, for a channel with average power constraint $P_{0}$, is
\begin{equation}
\mathbb{P} \left[ \mathbf{h} < \frac{2^{R_{0}}-1}{P_{0}/N_{0}} \right].
\label{outage}
\end{equation}

\begin{figure}[t]
\begin{center}
\includegraphics[width=3.3in]{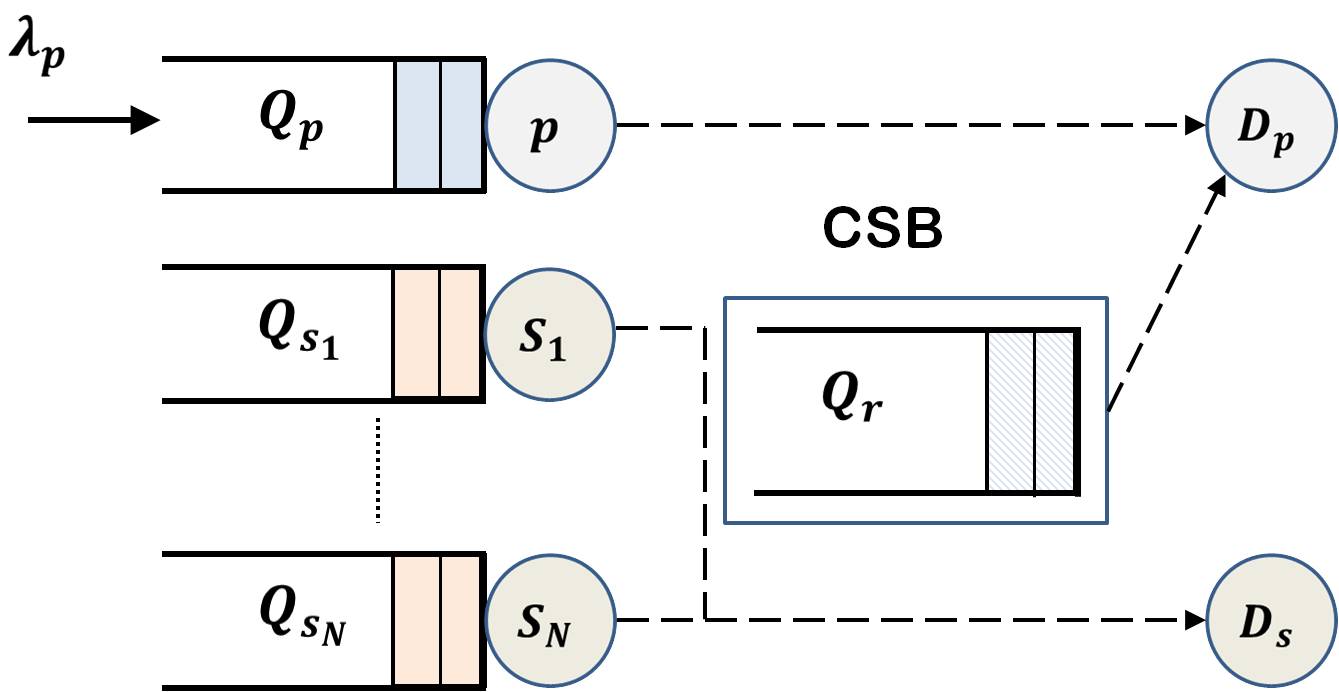}
\caption{Cognitive radio network model under consideration. The (logical) CSB is shown to coordinate the activities of the common relay queue.} \label{Fig1}
\end{center}
\end{figure}

\section{System Model}\label{system_model}
We consider the cognitive radio system shown in Fig. \ref{Fig1}.
The system comprises a PU $p$ that transmits its packets to a primary destination $D_{p}$. A cognitive network consisting of an arbitrary number of SUs coexists with the primary network. The number of SUs is denoted by $N$ and we refer to the set of SUs by $\mathbb{S}=\{ s_{i} \}_{i=1}^{N}$. Each SU has its own data that requires to be delivered to a common secondary destination $D_{s}$. All nodes are equipped with infinite capacity buffers. Time is slotted, and the transmission of a packet takes exactly one time slot. The duration of a time slot is normalized to unity and hence, the terms power and energy are used interchangeably in the sequel. We take into account the bursty nature of the source through modelling the arrivals at the PU as a Bernoulli process with rate $\lambda_{p}$ (packets/slot). In other words, at any given time slot, a packet
arrives at the PU with probability $\lambda_{p}<1$. The arrival process at the PU is independent and identically distributed (i.i.d.) across time slots. On the other hand, the SUs are assumed backlogged, i.e., SUs always have packets awaiting transmission. We assume that the SUs perfectly sense the PU's activity, i.e., there is no chance of collision between the PU and any of the secondary users. A node that successfully receives a packet broadcasts an acknowledgment (ACK) declaring the successful reception of that packet. ACKs sent by the destinations are assumed instantaneous and heard by all nodes error-free.

The channel between every transmitter-receiver pair exhibits frequency-flat Rayleigh block fading, i.e., the channel coefficient remains constant for one time slot and changes independently from one slot to another. The scalars $\mathbf{h}_{r_{i}}[n]$ and $\mathbf{h}_{s_{i}}[n]$ denote the absolute squared fading coefficient of the channels that connect the $i$th SU to $D_{p}$ and $D_{s}$, respectively, at the $n$th time slot. Similarly, the absolute squared fading coefficient of the channels that connect the PU to $D_{p}$ and $s_{i}$, at the $n$th time slot, are denoted by $\mathbf{h}_{p}[n]$ and $\mathbf{h}_{ps_{i}}[n]$, respectively. According to the Rayleigh fading assumption, $\mathbf{h}_{r_{i}}[n]$, $\mathbf{h}_{s_{i}}[n]$, and $\mathbf{h}_{ps_{i}}[n]$ are exponential random variables with means $\sigma^{2}$, for all $i=1,\hdots,N$. We denote an exponential random variable with mean $\sigma^{2}$ by $\mathrm{exp}(\sigma^{2})$. Then, we have $\mathbf{h}_{p}[n] \sim \mathrm{exp}(\sigma_{p}^{2})$. All links are considered statistically equivalent except for the link $p \rightarrow D_{p}$. We assume that $\sigma_{p}^{2}<\sigma^{2}$ to demonstrate the benefits of cooperation \cite{sadek}. For the ease of exposition, we set $\sigma^{2}=1$ throughout the paper. All communications are subject to additive white Gaussian noise of variance $N_{0}$.

Next, we present the queuing model of the system followed by the description of the employed cooperation strategy.

\subsection{Queuing Model}
The queues involved in the system analysis, shown in Fig. \ref{Fig1}, are described as follows:
\begin{itemize}
\item $Q_{p}$: a queue that stores the packets of the PU corresponding to the external Bernoulli arrival process with rate $\lambda_{p}$.

\item $Q_{s_{i}}$: a queue that stores the packets at the $i$th SU, where $i \in \{1,\hdots,N\}$.

\item $Q_{r}$: a queue that stores PU packets to be relayed to $D_{p}$.
\end{itemize}
Having independent relay queues for all SUs makes exact performance analysis intractable with the increasing number of users. To address this complexity, Krikidis \emph{et al.} introduced the idea of a common 'fictitious' relay queue $Q_{r}$ in \cite{Krikidis}, which is maintained by a so-called cluster supervision block (CSB) that controls and synchronizes all the activities of the cognitive cluster.
Along the lines of  \cite{Krikidis}, we assume the existence of a common relay such that SUs can perfectly exchange information with the CSB with a negligible overhead. The channels $\mathbb{S} \rightarrow D_{p},D_{s}$ are assumed known instantaneously at the CSB \cite{Krikidis,jovicic}.

The instantaneous evolution of queue lengths is captured as
\begin{align}\label{queue evolution}
\mathbf{Q}_{i}[n+1]= \left(\mathbf{Q}_{i}[n] - \mathbf{L}_{i}[n]\right)^{+} + \mathbf{A}_{i}[n], ~i \in \{p,r\} \cup \mathbb{S}
\end{align}
where $(x)^{+}=\text{max}(x,0)$ and $\mathbf{Q}_{i}[n]$ denotes the number of packets in the $i$th queue at the beginning of the $n$th time slot. The
binary random variables taking values either $0$ or $1$, $\mathbf{L}_{i}[n]$ and $\mathbf{A}_{i}[n]$, denote the departures and arrivals corresponding to the
$i$th queue in the $n$th time slot, respectively.

\subsection{Cooperation Strategy}\label{cooperation strategy}
The employed cooperative scheme is described as follows.
\begin{enumerate}
\item The PU transmits a packet whenever $Q_{p}$ is non-empty.

\item If the packet is successfully decoded by $D_{p}$, it broadcasts an ACK and the packet is dropped from $Q_{p}$.

\item If the packet is not successfully received by $D_{p}$ yet successfully decoded by at least one SU, an ACK is broadcasted and the packet is buffered in $Q_{r}$ and dropped from $Q_{p}$.

\item If $D_{p}$ and $\mathbb{S}$ fail to decode the packet, it is kept at $Q_{p}$ for retransmission in the next time slot.

\item When the PU is sensed idle, if $Q_{r}$ is non-empty, two out of all SUs transmit simultaneously. One SU is selected to relay a packet from $Q_{r}$ to $D_{p}$ and is denoted by $r^{*}$. Another SU is selected to transmit a packet of its own to $D_{s}$ and is denoted by $s^{*}$. Otherwise, if $Q_{r}$ is empty, one SU is selected to transmit a packet to $D_{s}$\footnote{Note that two SUs can be selected for transmission if $Q_r$ is empty. However, this requires multi-packet reception capability at the secondary destination which is out of the scope of this paper.}. The SUs' selection policies are explained in Section \ref{node_selection}.  

\item If the packets transmitted by the SUs are successfully received by their respective destinations, ACKs are broadcasted and these packets
    exit the system. Otherwise, the packet that experiences unsuccessful transmission is kept at its queue for later retransmission.

\end{enumerate}

\section{Power Allocation and Node Selection}\label{CS_PA}
In this section, we introduce the adaptive power allocation and opportunistic relay selection strategies for an arbitrary number of SUs, $N \geq 2$. We propose a power allocation policy
that minimizes energy consumption at each SU as compared to a fixed power allocation policy in \cite{Krikidis}.
In the sequel, node selection policy refers to the choice of the SU that relays a primary packet from $Q_{r}$ to $D_{p}$, and the SU that transmits a packet from its own queue to $D_{s}$, i.e., the selection of $r^*$ and $s^*$. The availability of CSI for all the channels (and thereby incurred interference) at the CSB is exploited to perform power allocation and node selection online, i.e., every time slot.

\subsection{Power Allocation}
Whenever $Q_{p}$ is non-empty, the PU transmits a packet with average power $P_{0}$. However, when the PU is idle and $Q_{r}$ is non-empty, two SUs out of $N$ transmit simultaneously by employing DPC \cite{DPC}. One SU relays a primary packet to $D_{p}$ while the other transmits a secondary packet to $D_{s}$. Since all SUs can perfectly exchange information with the CSB, $Q_{r}$ is accessible by both SUs selected for transmission. Therefore, the transmission of $r^*$ is considered a priori known interference at $s^*$. Accordingly, $s^*$ adapts its signal to see an interference-free link to $D_s$ using the result stated in Section \ref{dpc}. On the other hand, $s^*$ transmits a packet from its own queue which is not accessible by $r^*$. Thus, the transmission of $s^*$ causes an interference on the relay link, i.e., $r^* \rightarrow D_{p}$. The achievable rate region on this $Z$-interference channel at the $n$th time slot is given by
\begin{eqnarray}
\mathbf{R}_{s^*}[n]&=& \log \left[ 1 + \frac{P_{s^*}[n]\mathbf{h}_{s^*}[n]}{N_{0}} \right] \label{R1} \\
\mathbf{R}_{r^*}[n]&=& \log \left[ 1 + \frac{P_{r^*}[n]\mathbf{h}_{r^*}[n]}
{N_{0}+ P_{s^*}[n]\mathbf{h}_{\mathrm{I}}[n]} \right] \label{R2}
\end{eqnarray}
where $P_{s^*}[n]$ and $P_{r^*}[n]$ denote the instantaneous transmit powers of $s^*$ and $r^*$, respectively. The instantaneous absolute squared fading coefficients of the secondary, relay, and interference links are denoted by $\mathbf{h}_{s^*}[n]$, $\mathbf{h}_{r^*}[n]$, and $\mathbf{h}_{\mathrm{I}}[n]$, respectively. We denote the links $s^* \to D_{s}$, $r^* \to D_{p}$, and $s^* \to D_{p}$ by the secondary, the relay, and the interference link, respectively. 
Hereafter, we omit the temporal index $n$ for simplicity. Nevertheless, it is implicitly understood that power allocation and node selection are done on a slot-by-slot basis. In this work, we focus on developing an adaptive power allocation scheme for the transmitting SUs that use a fixed transmission rate $R_{0}$. Specifically, our multi-criterion objective is to enhance primary and secondary throughput while minimizing the energy consumption at each SU. The rates given by (\ref{R1}) and (\ref{R2}) stimulate thinking about how power is allocated to both transmitting SUs.

Next, we investigate two different power allocation policies for the SUs, namely, equal power (EP) allocation and adaptive power (AP) allocation. It is worth noting that power allocation and node selection are performed for the SUs since we have no control on the PU. 
Thus, in the following lines, we focus on the slots in which the PU is idle.

\subsubsection{Equal Power Allocation}\label{UP}
This policy assigns equal transmission powers to the SUs as proposed in \cite{Krikidis} and serves as a baseline scheme in this work. Whenever an SU transmits, it uses an average power $P_{\mathrm{max}}$. Specifically, if an SU is transmitting alone, e.g., $Q_{r}$ is empty, it uses a power $P_{\mathrm{max}}$. If two SUs transmit simultaneously, e.g., $Q_{r}$ is non-empty, $P_{s^*}=P_{r^*}=P_{\mathrm{max}}$.

\subsubsection{Adaptive Power Allocation}
\label{AP}
Unlike EP allocation, we exploit the CSI available at the CSB to propose an AP allocation scheme that minimizes the average power consumption
at each SU. We use (\ref{R1}) and (\ref{R2}) along with (\ref{outage}) to derive conditions on $P_{s^*}$ and $P_{r^*}$ for successful transmission at a target transmission rate $R_{0}$.
These conditions are
\begin{eqnarray}
P_{s^*} &\geq & \frac{(2^{R_{0}}-1)N_{0}}{\mathbf{h}_{s^*}} \label{P1} \\
P_{r^*} &\geq & \frac{(2^{R_{0}}-1)[N_{0}+ P_{s^*}\mathbf{h}_{\mathrm{I}}]}
{\mathbf{h}_{r^*}}. \label{P2}
\end{eqnarray}
A transmitter that violates the condition on its transmission power experiences a sure outage event. Furthermore, we impose a maximum power constraint at each SU, where $P_{s^*},P_{r^*} \leq P_{\mathrm{max}}$. It is worth noting that $P_{s^*}$ is computed first according to (\ref{P1}) followed by the computation of $P_{r^*}$ according to (\ref{P2}). In a given slot, if $P_{\mathrm{max}}$ is less
than the power required to guarantee a successful transmission for a given SU, i.e., $P_{\mathrm{max}}$ is less than the right hand sides of either (\ref{P1}) or (\ref{P2}), the CSB sets the power of that SU to zero to avoid a guaranteed outage event. Clearly, this results in increasing the throughput of the PU due to reduction in the amount of interference caused by the transmission of $s^*$ on the relay link in the time slots where $s^*$ refrains from transmitting. Moreover, compared to EP allocation, energy wasted in slots where a sure outage event occurs is now saved.

\subsection{Node Selection Policies}\label{node_selection}
We consider a system that assigns full priority to the PU to transmit whenever it has packets. Therefore, the SUs continuously monitor the PU's
activity seeking an idle time slot. When the PU is sensed idle, the SUs are allowed to transmit their own and/or a packet from the common queue $Q_r$. Note that it is possible to transmit only one packet by the SUs in the following scenarios:
\begin{enumerate}
  \item If $Q_{r}$ is empty, i.e., no primary packet to be relayed. Then, we select the SU with the best channel to $D_s$.
  \item $Q_{r}$ is non-empty, but $r^*$ or $s^*$ is set silent by the CSB to avoid a guaranteed outage event on the $r^*\to D_p$ or $s^*\to D_s$ link. Note that CSI for transmission is assumed to be known at CSB and outage event (due to power limitation) can be predicted before transmission as discussed in Section \ref{AP}. 
In this case, we choose the transmitting SU as the one with the best instantaneous link to the intended destination. For example, if $r^*$ is silent and $s^*$ is transmitting alone, the SU with the best link
    between $\mathbb{S} \to D_{s}$ transmits.
\end{enumerate}
The case for the simultaneous transmission of two SUs is the main topic for investigation in this paper. If the two transmissions occur simultaneously, the transmitting SUs are selected according to one of the following policies.

\subsubsection{Best secondary link (BSL)}\label{BSL}
In this policy, the utility function to be maximized is the SU throughput. Therefore,
we choose the SU that transmits a packet of its own as the one with the best instantaneous link to $D_{s}$, i.e.,
\begin{equation}
\mathbf{h}_{s^*}=\underset{i \in \{1,\hdots,N\}}{\mathrm{max.}}~\mathbf{h}_{s_{i}}.
\end{equation}
Among the remaining $(N-1)$ SUs, the one with the best instantaneous link to $D_{p}$ is chosen to be $r^*$.

\subsubsection{Best primary link (BPL)}\label{BPL}
In this policy, unlike BSL, the utility function to be maximized is PU throughput. Thus,
we choose the SU that relays a primary packet from $Q_{r}$ as the one with the best instantaneous link to $D_{p}$, i.e.,
\begin{equation}
\mathbf{h}_{r^*}=\underset{i \in \{1,\hdots,N\}}{\mathrm{max.}}~\mathbf{h}_{r_{i}}.
\end{equation}
Among the remaining $(N-1)$ SUs, the one with the best instantaneous link to $D_{s}$ is chosen to be $s^*$.

It is worth noting that all links $\mathbb{S} \rightarrow D_{p},D_{s}$ are statistically independent. Thus, at any given time slot, if a certain
SU has the best instantaneous channel to a certain destination, e.g., $D_{p}$, we can not infer any information about its link quality to the
other destination, e.g., $D_{s}$. Hence, $\forall i \in \{1,\hdots,N\}$, $s_{i}$ can have the best link to $D_{p}/D_{s}$ irrespective of the quality of
its link to the other destination.

So far, we have introduced two policies for each of the power allocation and SU scheduling policies. Thus, we have four different cases
arising from the possible combinations of these policies. Next, we proceed with the performance analysis of the system for each case.

\section{Throughput and Delay Analysis}
\label{sect:analysis}
In this section, we conduct a detailed analysis for the system performance in terms of throughput and delay. Towards this objective, we derive
the stability conditions on the queues with stochastic packet arrivals, namely, $Q_{p}$ and $Q_{r}$. The stability of a queue is loosely defined as
having a bounded queue size, i.e., the number of packets in the queue does not grow to infinity \cite{sadek}. Furthermore, we analyze the average queuing delay of the primary packets. We obtain a closed-form expression for this delay through deriving the moment generating function (MGF) of the joint lengths of $Q_{p}$ and $Q_{r}$. It is worth noting that the SUs' queues are assumed backlogged and hence, no queueing delay analysis is performed for the secondary packets. In the following lines, we provide a general result for the throughput of the primary and secondary links as well as the delay of primary packets.
Then, we proceed to highlight the role of the proposed power allocation and node selection policies.
We first introduce some notation. The probabilities of successful transmissions on the relay and secondary links are denoted by $f_{r^*}$ and $f_{s^*}$, respectively. A transmission on the link $p \to D_{p}$ is successful with probability $f_{p}$. In addition, the probability that at least one SU successfully decodes a transmitted primary packet is denoted by $f_{ps}$. 

\begin{theorem}\label{Thm1}
The maximum achievable PU throughput for the system shown in Fig. \ref{Fig1}, under any combination of  power allocation and node selection
policies, is given by
\begin{equation}\label{PU_thrpt}
\lambda_{p}<\frac{f_{r^*}[f_{p}+(1-f_{p})f_{ps}]}{f_{r^*}+(1-f_{p})f_{ps}}
\end{equation}
while the throughput of the SU $s_{i} \in \mathbb{S}$ is given by
\begin{equation}\label{SUthrpt}
\mu_{s_i}=\frac{1}{N} \left[1-\frac{\lambda_p}{f_{p}+(1-f_{p})f_{ps}}\right]f_{s^*}.
\end{equation}
\end{theorem}
\begin{proof}
We use Loynes' theorem \cite{Loynes} to establish the stability conditions for $Q_{p}$ and $Q_{r}$. The theorem states that if the arrival and
service processes of a queue are stationary, then the queue is stable if and only if the arrival rate is strictly less than the service rate.
Therefore, for $Q_{p}$ to be stable, the following condition must be satisfied
\begin{equation}\label{Qp_stability}
\lambda_{p}<\mu_{p}
\end{equation}
where $\mu_{p}$ denotes the service rate of $Q_{p}$. A packet departs $Q_{p}$ if it is successfully decoded by at least one node in $\mathbb{S} \cup \{D_{p}\}$. Thus, $\mu_{p}$ is given by
\begin{align}\label{mu_p}
\mu_{p}=f_{p} + (1-f_{p})f_{ps}.
\end{align}
Similarly, $Q_{r}$ is stable if
\begin{align}\label{Qr}
\frac{\lambda_{p}}{\mu_{p}}(1-f_{p})f_{ps}<\left[ 1-\frac{\lambda_{p}}{\mu_{p}}\right]f_{r^*}.
\end{align}
A PU's packet arrives at $Q_{r}$ if $Q_{p}$ is non-empty and an outage occurs on the direct link $p \to D_p$ yet no outage occurs at least on one link between $p \rightarrow \mathbb{S}$. From Little's theorem \cite{Bertsekas}, we know that probability of $Q_{p}$ being non-empty equals
$\lambda_{p}/\mu_{p}$. This explains the rate of packet arrivals at $Q_{r}$ shown on the left hand side (LHS) of (\ref{Qr}). The right hand side (RHS) represents the service rate of $Q_{r}$. A packet departs $Q_{r}$ if $Q_{p}$ is empty and there is no outage on the link $r^* \rightarrow D_{p}$.
Rearranging the terms of (\ref{Qr}), we obtain the maximum achievable PU throughput as given by (\ref{PU_thrpt}) provided that $\mu_{p}$ is given by (\ref{mu_p}). It is worth noting that (\ref{PU_thrpt}) provides a tighter bound on $\lambda_{p}$ than (\ref{Qp_stability}) due to the
multiplication of $\mu_{p}$ in (\ref{PU_thrpt}) by a term less than one.

On the other hand, we compute the throughput of SUs by calculating the service rate of their queues since they are assumed backlogged. Due to the symmetric configuration considered, i.e., statistically equivalent links $\mathbb{S} \rightarrow D_{s}$, the throughput of all SUs is the same. For $s_i \in \mathbb{S}$, a packet departs $Q_{s_i}$ if $Q_{p}$ is empty, $s_{i}$ is selected to transmit a packet of its own and no outage occurs on the link $s_{i} \rightarrow D_{s}$. Due to symmetry, at any time slot, all SUs have equal probabilities to be selected to transmit a packet from their own queues, i.e., $\mathbb{P}[s^*=s_i]=1/N ~ \forall i \in \{1,\hdots,N\}$. Therefore, the SUs' throughput is given by (\ref{SUthrpt}) provided that $\mu_{p}$ is given by (\ref{mu_p}).
\end{proof}

Next, we develop a mathematical framework to analyze the average queuing delay for the PU's packets.

\begin{theorem}\label{Thm2}
The average queuing delay encountered by the PU packets in the system shown in Fig. \ref{Fig1}, under any combination of power
allocation and node selection policies, is
\begin{equation}\label{tau}
\tau= \frac{N_{p} + N_{r}}{\lambda_{p}}
\end{equation}
where $N_{p}$ and $N_{r}$, the average lengths of $Q_{p}$ and $Q_{r}$, respectively, are given by
\begin{eqnarray}
N_{p}&=&\frac{-\lambda_{p}^{2} + \lambda_{p}}{\mu_{p}-\lambda_{p}}\label{pk} \\
N_{r}&=& \frac{r \lambda_{p}^{2} + s \lambda_{p}}{\delta \lambda_{p}^{2} + \zeta \lambda_{p} + \eta}\label{Nr}
\end{eqnarray}
and
\begin{eqnarray}
r&=&f_{ps}(1-f_{p}) \left[ \frac{f_{r^*}-f_{p}}{\mu_{p}}- f_{r^*}-f_{ps}(1-f_{p}) \right]\\
s &=& f_{ps}(1-f_{p})\mu_{p}\\
\delta &=& f_{r^*} + f_{ps}(1-f_{p})\\
\zeta &=& \mu_{p} \left[ -2f_{r^*}-f_{ps}(1-f_{p}) \right]\\
\eta &=& \mu_{p}^{2} f_{r^*}
\end{eqnarray}
while $\mu_{p}$ is given by (\ref{mu_p}).
\end{theorem}
\begin{proof}
If a primary packet is directly delivered to $D_{p}$, it experiences the queuing delay at $Q_{p}$ only.
This happens with a probability $1-\epsilon=f_{p}/\ \mu_{p}$. However, if the packet is forwarded to $D_{p}$ through the relay link, it experiences the total queuing delay at both $Q_{p}$ and $Q_{r}$. Thus, the average delay is
\begin{equation}\label{delay_averaging}
\tau= (1-\epsilon)\tau_{p} + \epsilon (\tau_{p}+ \tau_{r})= \tau_{p} + \epsilon \tau_{r}
\end{equation}
where $\tau_{p}$ and $\tau_{r}$ denote the average delays at $Q_{p}$ and $Q_{r}$, respectively.
The arrival rates at $Q_{p}$ and $Q_{r}$ are given by $\lambda_{p}$ and $\epsilon \lambda_{p}$, respectively. Thus, applying Little's law
\cite{Bertsekas} renders
\begin{equation}\label{delay_formula}
\tau_{p}=N_{p}/ \lambda_{p}, \hspace{10mm} \tau_{r}=N_{r}/ \epsilon \lambda_{p}.
\end{equation}
Substituting (\ref{delay_formula}) in (\ref{delay_averaging}) renders $\tau$ exactly matching (\ref{tau}).

Proceeding with computing $N_{p}$, we make use of the fact that $Q_{p}$ is a discrete-time $M/M/1$ queue with arrival rate $\lambda_{p}$ and
service rate $\mu_{p}$. Thus, $N_{p}$ is directly given by (\ref{pk}) through applying the Pollaczek-Khinchine formula \cite{PK}. However, the
dependence of the arrival and service processes of $Q_{r}$ on the state of $Q_{p}$ necessitates using a MGF approach \cite{sidi} to calculate
$N_{r}$. The MGF of the joint lengths of $Q_{p}$ and $Q_{r}$ is defined as
\begin{equation}
G(x,y)= \lim_{n \rightarrow \infty} \mathbb{E} \left[ x^{\mathbf{Q}_{p}[n]} y^{\mathbf{Q}_{r}[n]} \right]
\end{equation}
where $\mathbb{E}$ denotes the statistical expectation operator. Following the framework in \cite{Ephremedis,ashour2015cognitive}, we get
\begin{equation}\label{Gxy}
G(x,y)=(\lambda_{p}x+1-\lambda_{p})\frac{B(x,y)G(0,0)+C(x,y)G(0,y)}{yD(x,y)}
\end{equation}
where
\begin{align}
B(x,y) &=x(y-1)f_{r^*} \nonumber \\
C(x,y) &=xf_{r^*}- yf_{p}- y^{2}f_{ps}(1- f_{p})+ xy(\mu_{p}- f_{r^*}) \nonumber \\
D(x,y) &=x \! - \!(\lambda_{p}x \!+\!  1  \!-\!  \lambda_{p})[f_{p}\! +\!  yf_{ps}(1 -  f_{p})+  x(1 -  \mu_p)].
\end{align}
First, we compute the derivative of (\ref{Gxy}) with respect to $y$ and then, take the limit of the result
when $x$ and $y$ tend to $1$. This verifies that $N_{r}$ is given by (\ref{Nr}).
\end{proof}

Theorems \ref{Thm1} and \ref{Thm2} provide closed-form expressions for the network performance metrics, throughput and delay. These expressions are mainly functions of the outage probabilities on various links in the network, namely, $f_{p}$, $f_{ps}$, $f_{r^*}$, and $f_{s^*}$. In the following lines, we quantify these outage probabilities for the different combinations of power allocation and node selection policies. It is worth noting that $f_{p}$ and $f_{ps}$ are related to the PU side. Therefore, they remain the same for all combinations of power allocation and node selection policies which are performed at the SUs side. Using (\ref{outage}), we have
\begin{align}\label{fpDp}
f_{p}=\mathbb{P}\left[ \mathbf{h}_{p} > \frac{2^{R_{0}}-1}{P_{0}/N_{0}} \right]=e^{-\alpha/\sigma_{p}^{2}}
\end{align}
where $\alpha=\frac{2^{R_{0}}-1}{P_{0}/N_{0}}$. This follows from the Rayleigh fading assumption that renders $\mathbf{h}_{p} \sim \mathrm{exp}(\sigma_{p}^{2})$. Similarly,
\begin{align}\label{fpstar}
f_{ps}=\mathbb{P}\left[ \underset{i \in \{1,\hdots,N\}}{\mathrm{max}} \mathbf{h}_{ps_i} > \alpha \right]=1-(1-e^{-\alpha})^N.
\end{align}

On the other hand, we shift our attention to the SU side to calculate $f_{r^*}$ and $f_{s^*}$. We analyze the four cases arising from the proposed power allocation and relay selection policies in the following order: (i) EP-BSL, (ii) EP-BPL, (iii) AP-BSL, and (iv) AP-BPL. 
Towards this objective, we first note that for each SU, its link qualities to $D_{p}$ and $D_{s}$ are statistically independent. Furthermore, these links are independent of the other $(N-1)$ users' links. Thus, we are dealing with $2N$ i.i.d. random variables, $\mathbf{h}_{r_{i}}$ and $\mathbf{h}_{s_{i}}$, $\forall i \in \{1,\hdots,N\}$. Each of these variables is exponentially distributed with mean $1$ as a direct consequence of the Rayleigh fading model considered. We begin with an analysis of the distributions of the random variables involved in the derivations of $f_{r^*}$ and $f_{s^*}$, specifically, $\mathbf{h}_{r^*}$, $\mathbf{h}_{\mathrm{I}}$, and $\mathbf{h}_{s^*}$. Finding these distributions is fundamental to the mathematical derivations presented next. Obviously, the distributions is dependent on the node selection policy employed and hence, we present a separate analysis for BSL and BPL in Appendices \ref{BSL_dist} and \ref{BPL_dist}, respectively.

For the ease of exposition, we define $a=\frac{2^{R_{0}}-1}{P_{\mathrm{max}}/N_{0}}$, $b=(2^{R_{0}}-1)^{-1}$, and $\beta=1-e^{-a}$. The exponential integral function, $E_{1}[.]$, is defined as $E_{1} [x]=\int_{x}^{\infty}(e^{-t}/t)dt$.
\begin{lemma}\label{lemma2}
For EP-BSL, $f_{r^*}$ and $f_{s^*}$ are given by
\begin{eqnarray}
f_{r^*}&=& 1- \displaystyle \sum_{k=0}^{N-1} {N-1 \choose k} (-1)^{k}
\frac{e^{-ka}}{(1+k/b)} \label{fkpDpBSL} \\
f_{s^*}&=&1-\beta^N. \label{fksDsBSL}
\end{eqnarray}
\end{lemma}
\begin{proof}
See Appendix \ref{EP-BSL}.
\end{proof}

\begin{lemma}\label{lemma3}
For EP-BPL, $f_{r^*}$ is given by
\begin{equation}
f_{r^*}=\frac{N}{N-1} \displaystyle \sum_{k=1}^{N-1} {N-1 \choose k-1} 
\left[ \mathrm{I}_{1} - \mathrm{I}_{2} \right]. \label{fskpDp_FPBPL}
\end{equation} 
where
\begin{align}
&\mathrm{I}_{1}=\displaystyle \sum_{m=0}^{k-1} {k-1 \choose m} \frac{(-1)^{m}}{(N-k+m+1)} \\
&\mathrm{I}_{2}=\displaystyle \sum_{m=0}^{k-1} \sum_{\ell=0}^{N} {k-1 \choose m} {N \choose \ell}
\frac{(-1)^{m+\ell}e^{-a \ell}}{(N-k+m+\ell / b +1)} 
\end{align}
On the other hand, $f_{s^*}$ is given by
\begin{equation}
f_{s^*}=\gamma \left(1-\beta^{N-1}\right) + (1-\gamma)\left(1-\beta^N\right) \label{fksDs_UPBPL}
\end{equation}
where
\begin{equation}\label{gamma}
\gamma=\frac{\lambda_{p}(1-f_{p})f_{ps}}{(\mu_{p}-\lambda_{p})f_{r^*}}.
\end{equation}
\end{lemma}
\begin{proof}
See Appendix \ref{EP-BPL}.
\end{proof}

\begin{lemma}\label{lemma4}
For AP-BSL, $f_{r^*}$ is given by
\begin{equation}\label{fkpDpAPBSL}
f_{r^*}\!\!=\!\!\beta^N\!(1\!-\!\beta^N\!)\!+ 
N \!\!\displaystyle \sum_{k=0}^{N-1}\!\! {\!N\!-\!1 \choose k\!} \!(\!-1\!)^{k} e^{-a(k+1)}\! 
\left[ \mathrm{I}_{3}\! -\! \mathrm{I}_{4} \right] \quad
\end{equation}
where
\begin{align}
& \mathrm{I}_{3}=\frac{N-1}{k+1} \displaystyle \sum_{\ell=0}^{N-2} {N-2 \choose \ell} 
\frac{(-1)^{\ell} }{{(\ell + 1})}  e^{-a(\ell+1)}  \label{I} \\ 
& 
\mathrm{I}_{4}= 
\frac{a}{b}e^{ab}(N-1)\! 
\displaystyle \sum_{\ell=0}^{N-2} \!\!{N-2 \choose \ell}\! (-1)^{\ell}
e^{\frac{a(1+b+\ell)(k+1-b)}{b}} \notag \\
& \phantom{= \frac{a}{b}e^{ab}(N-1)\! \displaystyle \sum_{\ell=0}^{N-2} x}
\times E_{1} \left[ \frac{a(1+b+\ell)(k+1)}{b} \right]. \label{II}
\end{align}
On the other hand, $f_{s^*}$ is given by (\ref{fksDsBSL}).
\end{lemma}
\begin{proof}
See Appendix \ref{AP-BSL}.
\end{proof}


\begin{figure*}[t]
 \centering
 \subfigure[AP-BSL.]
 {\includegraphics[width=1\columnwidth , height=0.65\columnwidth]{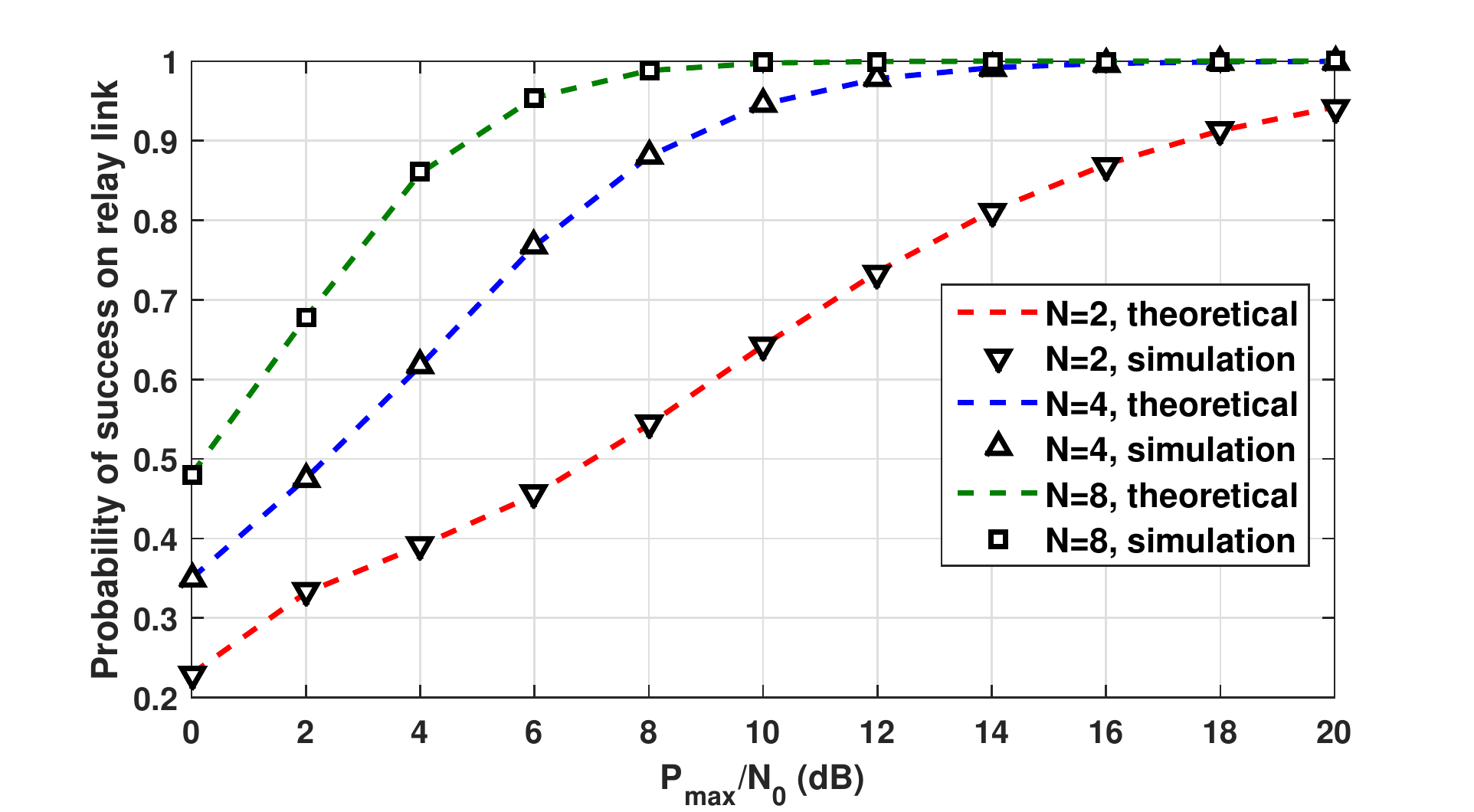}\label{Fig2}}
\subfigure[AP-BPL.]
 {\includegraphics[width=1\columnwidth , height=0.65\columnwidth]{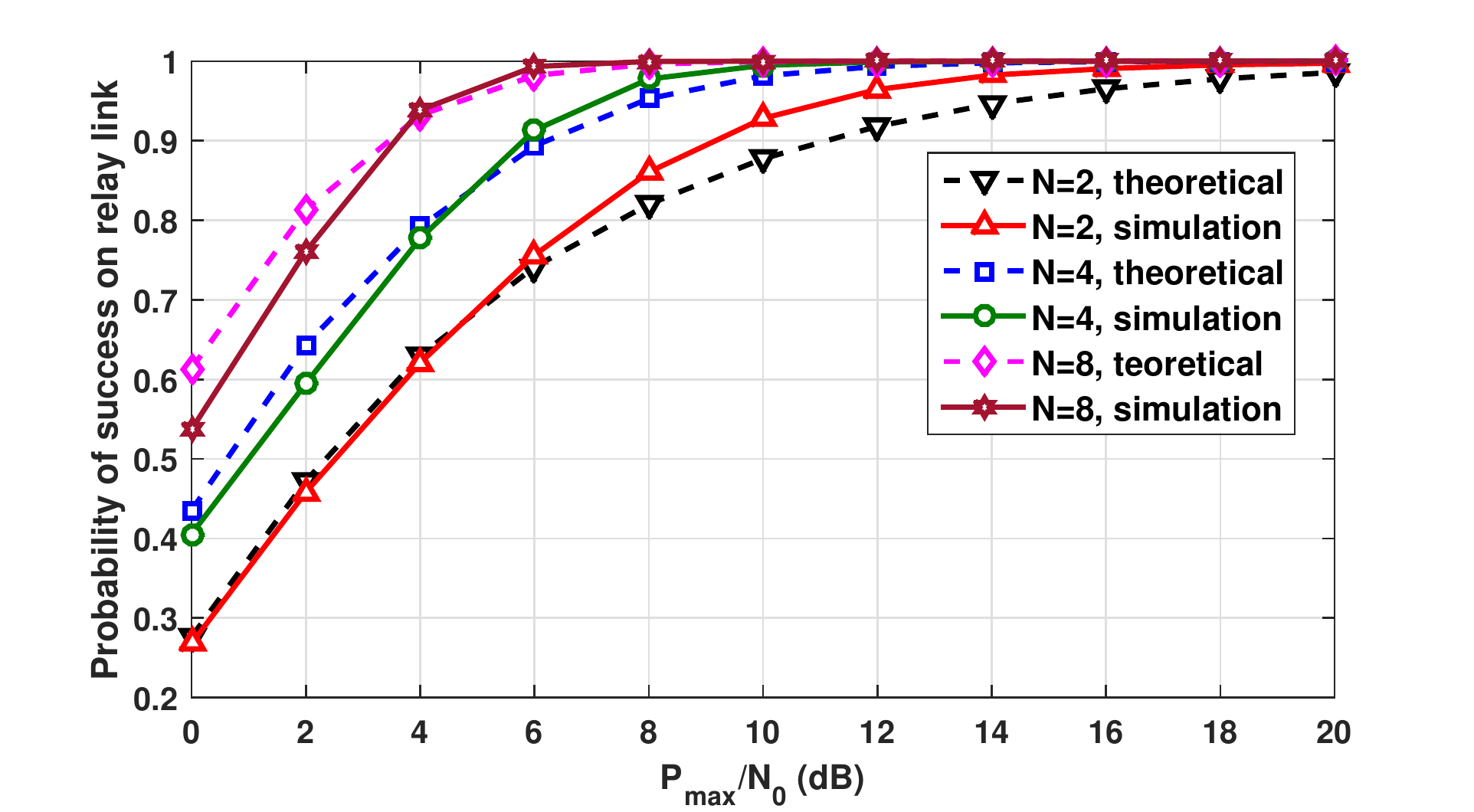}\label{Fig3}}
 \caption{The probability of transmission success on the relay link versus $P_{\mathrm{max}}/N_{0}$ for AP-based schemes.}
 \label{fr_sim_theo}
  \vspace{-3mm}
\end{figure*}

\begin{lemma}\label{lemma5}
For AP-BPL, $f_{r^*}$ is given by
\begin{align} \label{fr_APBPL}
f_{r^*}\!\!=\!\!\!  &
\displaystyle \sum_{k=1}^{N-1} \sum_{\ell=0}^{k-1} \sum_{m=0}^{N-2}\!\!  
{\!N\!-\!1\! \choose\! k\!-\!1\!}\! {\!k\!-\!1\! \choose\! \ell\!} \!{\!N\!-\!2\! \choose\! m\!} 
\!
\frac{(-\!1)^{m+\ell}N^2\left[ \mathrm{I}_{5}\! -\! \mathrm{I}_{6} \right]}{(N\!-\!k\!+\!\ell\!+\!1)}
\notag \\
& +
\beta^{N-1}(1-\beta^{N})  
\end{align}
where
\begin{align}
& \mathrm{I}_{5}=
\frac{e^{-a(m+1)}}{(m+1)}
\displaystyle \sum_{n=0}^{N-1}
{N-1 \choose n} \frac{(-1)^{n} e^{-a(n+1)}}{(n+1)} \label{I5}  \\
& 
\mathrm{I}_{6}\!=\!\! 
\displaystyle \sum_{n=0}^{N-1}\!\! {\!N\!-\!1 \choose\! n\!}\!   
\frac{a(-1)^{n} e^{-a(m+n-t+2)}}{(t-n-1)}
e^{tc} E_{1}\!\left[ t(a+c) \right] \label{I6}
\end{align}
and the terms $t$ and $c$ are 
\begin{align}
& t=b(N-k+\ell+1)+n+1 \label{t} \\
& c=a\left[ \frac{m+1}{b(N-k+\ell+1)}-1 \right]. \label{c}
\end{align}
On the other hand, $f_{s^*}$ is given by (\ref{fksDs_UPBPL}). 
\end{lemma}
\begin{proof}
See Appendix \ref{AP-BPL}.
\end{proof}

\section{Numerical Results}\label{results}
In this section, we validate the closed-form expressions derived in the paper via comparing theoretical and numerical simulation results. We investigate the system performance in terms of the primary and secondary throughput as well as the average primary packets' delay. In addition, we quantify the average power consumption at the SUs. Furthermore, we conduct performance comparisons between the four strategies resulting from the proposed power allocation and SU selection policies. Accordingly, we draw insights about the benefit of employing the proposed power allocation schemes. We set $P_{0}/N_{0}=10$ dB. Results are averaged over $10^6$ time slots.

Theorems \ref{Thm1} and \ref{Thm2} provide closed-form expressions for primary and secondary throughput as well as average queueing delay for primary packets. Generic expressions have been provided that work for any combination of power allocation and node selection policies. These expressions are functions of the probabilities of successful transmissions on relay and secondary links, i.e., $f_{r^*}$ and $f_{s^*}$. This fact has been thoroughly addressed in the appendices, where the four different power allocation and node selection policies have been analyzed. We start by validating our theoretical findings through simulations. Towards this objective, the analytical expressions for $f_{r^*}$, derived in Appendix \ref{AP-BSL} and \ref{AP-BPL}, are compared to their corresponding simulation results for both AP-BSL and AP-BPL in Fig. \ref{fr_sim_theo}. We set a target rate $R_{0}=1.5$ (bits/channel use) and we choose $\sigma_{p}^{2}=0.25$. Fig. \ref{Fig2} shows a perfect match of theoretical and simulation results for AP-BSL for any number of SUs, $N$. However, for AP-BPL, Fig. \ref{Fig3} shows a slight deviation between both results. This difference is attributed to the relaxation of the constraint that $\mathbf{h}_{\mathrm{I}} < \mathbf{h}_{r^*}$ in the derivation presented in Appendix \ref{AP-BPL}, where we treat $\mathbf{h}_{\mathrm{I}}$ and $\mathbf{h}_{r^*}$ as independent random variables. This constraint is an immediate consequence of the node selection policy presented in Section \ref{BPL}. The relaxation has been done for the sake of mathematical tractability. Nevertheless, Fig. \ref{Fig3} shows that the constraint relaxation has a minor effect on the obtained closed-form expression for $f_{r^*}$. This validates our theoretical findings. Fig. \ref{fr_sim_theo} show that $f_{r^*}$ consistently increases as the number of SUs increases for both AP-based schemes. This behavior is also true for EP-based schemes and is attributed to multi-user diversity gains obtained through increasing $N$.


\begin{figure*}[t]
 \centering
 \subfigure[Maximum achievable PU throughput versus $P_{\mathrm{max}}/N_{0}$.]
 {\includegraphics[width=1\columnwidth , height=0.65\columnwidth]{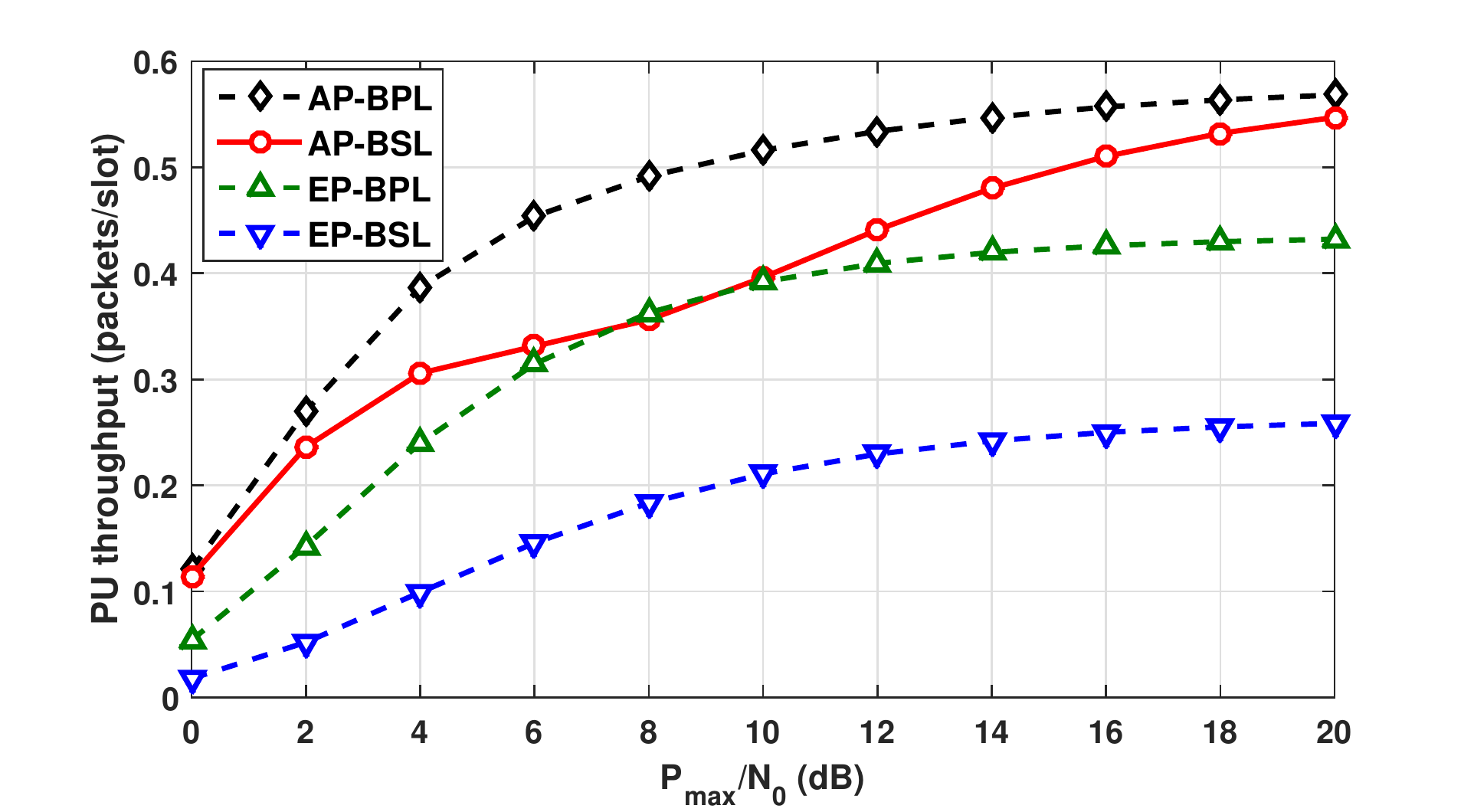}\label{Fig4}}
\subfigure[SU throughput versus $\lambda_{p}$.]
 {\includegraphics[width=1\columnwidth , height=0.65\columnwidth]{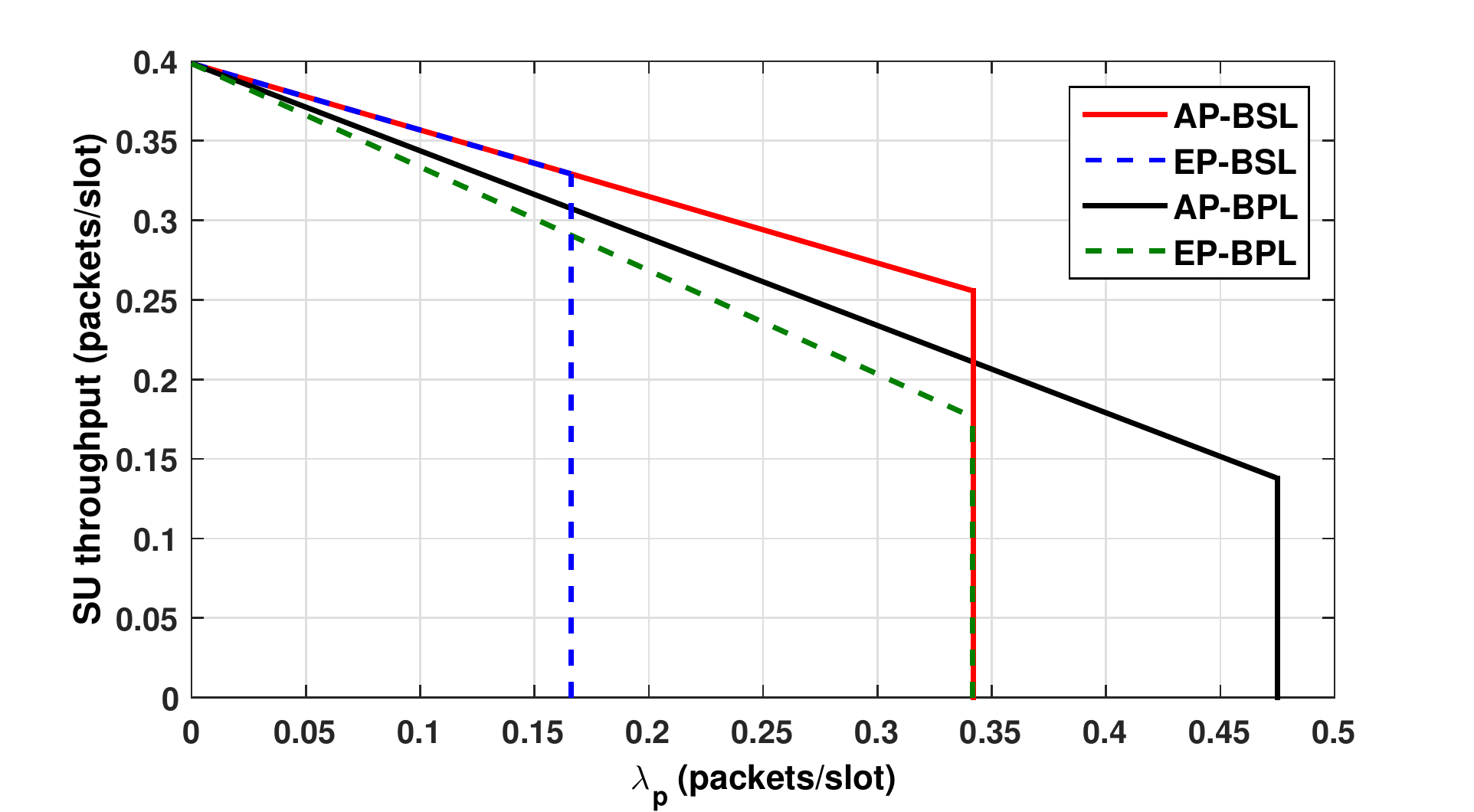}\label{Fig5}}
 \caption{The throughput of the PU and SUs for all combinations of power allocation and node selection policies.}
 \label{PU_SU_thrpt}
  \vspace{-3mm}
\end{figure*}


\begin{figure*}[t]
 \centering
 \subfigure[Average primary packets' delay versus $P_{\mathrm{max}}/N_{0}$.]
 {\includegraphics[width=1\columnwidth , height=0.65\columnwidth]{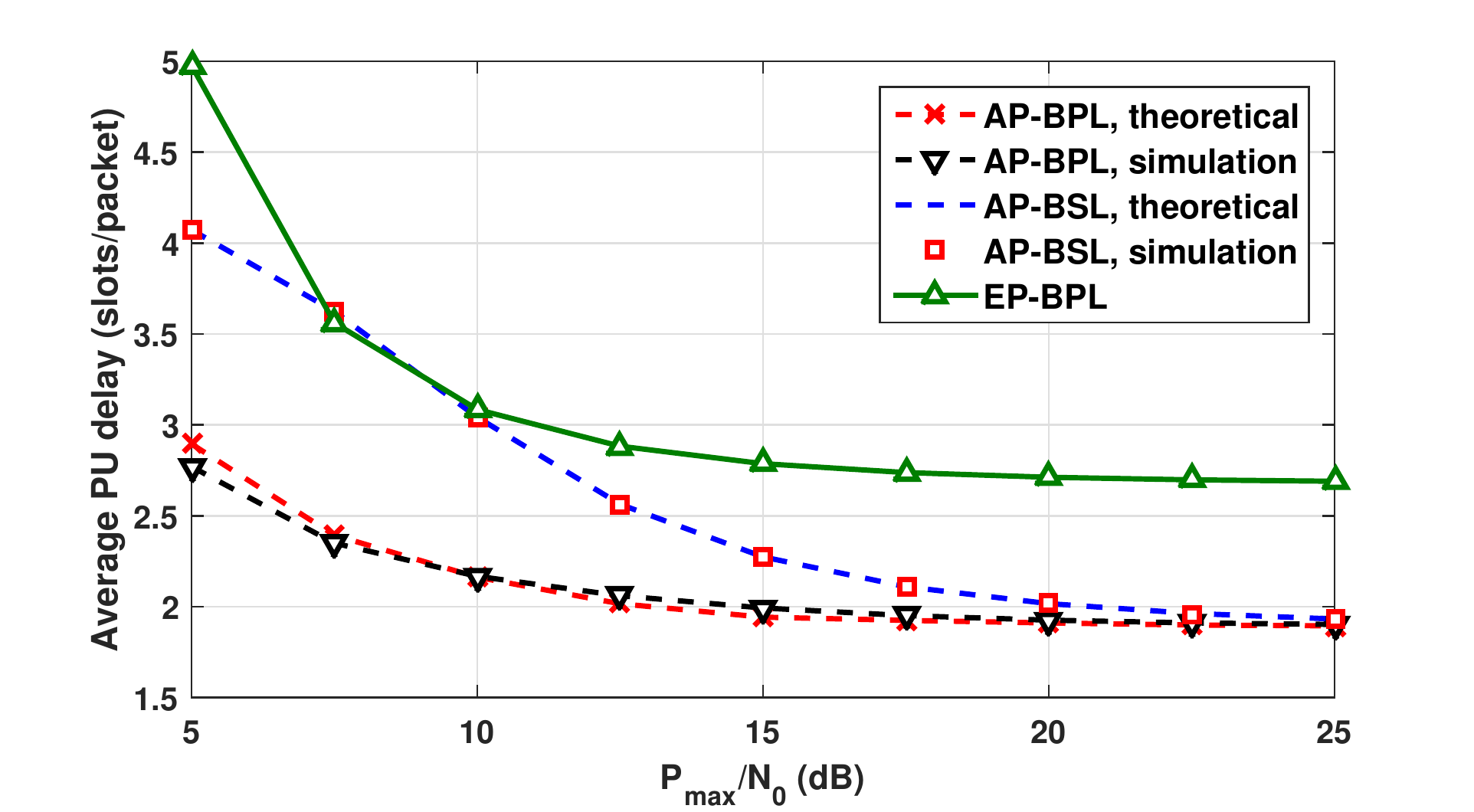}\label{Fig6}}
\subfigure[Average primary packets' delay versus $\lambda_{p}$.]
 {\includegraphics[width=1\columnwidth , height=0.65\columnwidth]{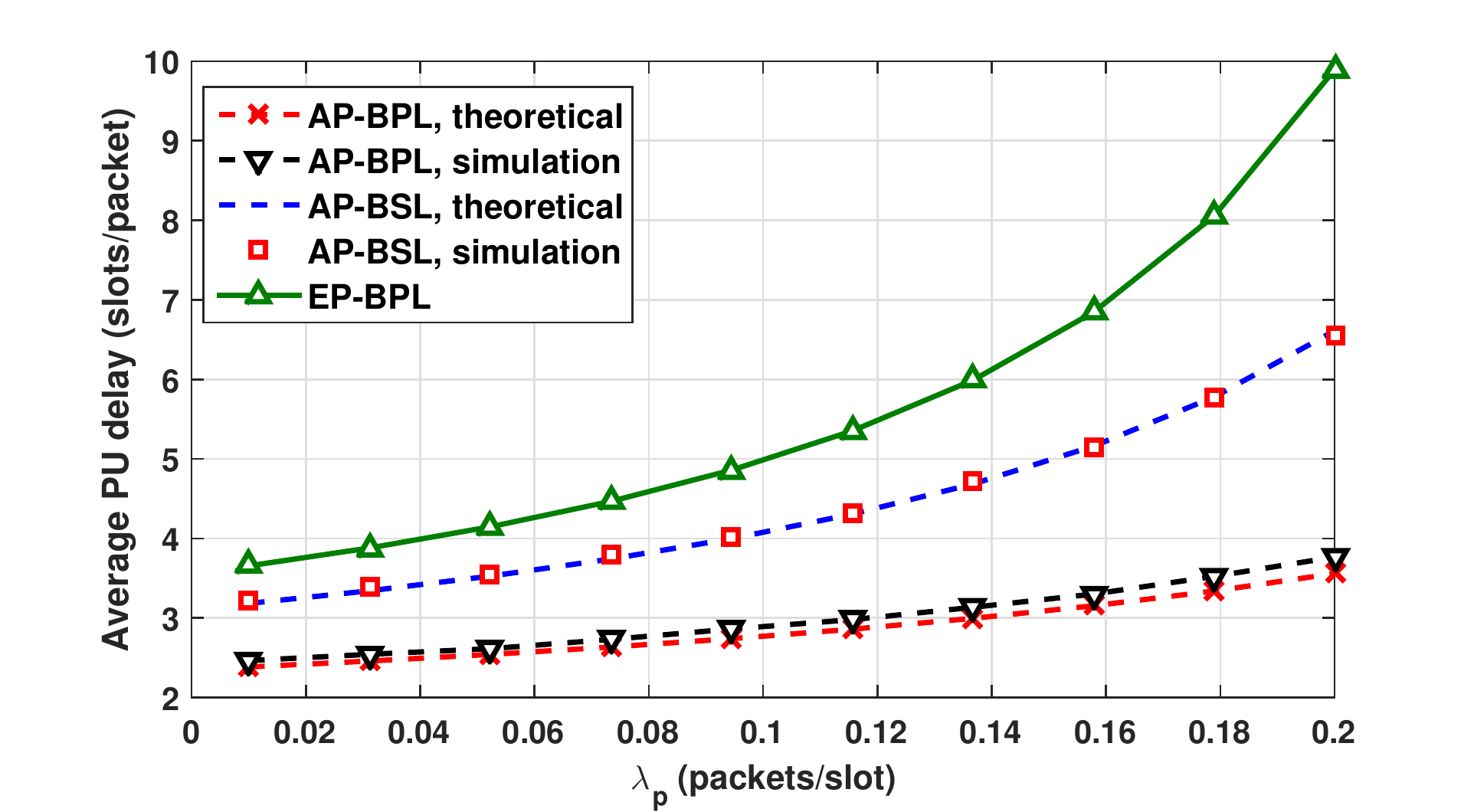}\label{Fig7}}
 \caption{The average queueing delay of PU's packets for different combinations of power allocation and node selection policies.}
 \label{PU_delay}
  \vspace{-3mm}
\end{figure*}

We investigated the effect of varying $N$ in Fig. \ref{fr_sim_theo}. Without loss of generality, the rest of the results are presented for $N=2$, $R_{0}=2$ (bits/channel use), and $\sigma_{p}^{2}=0.25$. We proceed with presenting the throughput of the PU and the SUs for all combinations of power allocation and node selection policies in Fig. \ref{PU_SU_thrpt}. In Fig. \ref{Fig4}, we plot the maximum achievable PU throughput, i.e., maximum achievable $\lambda_{p}$ given by (\ref{PU_thrpt}) in Theorem \ref{Thm1}, versus $P_{\mathrm{max}}/N_{0}$. AP-BPL is shown to outperform all other schemes. Moreover, it is evident that AP-based schemes outperform EP-based schemes \cite{Krikidis}, irrespective of the node selection policy employed. In Fig. \ref{Fig5}, we plot the SU throughput versus $\lambda_{p}$ at $P_{\mathrm{max}}/N_{0}=7$ dB. For the same node selection policy, the throughput region of the AP-based schemes is shown to strictly contain that of the EP based scheme. Furthermore, at every feasible $\lambda_{p}$ for EP-BPL, higher SU throughput is attained by AP-BPL. Thus, power adaptation expands the stable throughput region. This shows the superiority of AP-based schemes in both PU and SU throughput over their EP-based counterparts.

In Fig. \ref{PU_delay}, we study the average delay encountered by the PU packets. We refrain from plotting the results corresponding to EP-BSL to get a clear view of the comparison. EP-BSL yields much worse delay than the other three strategies. We plot the average primary packet delay versus $P_{\mathrm{max}}/N_{0}$ in Fig. \ref{Fig6}. As the available power resources increase, i.e., $P_{\mathrm{max}}/N_{0}$ increases, delay decreases. We attain lower average delay through power adaptation. As expected, AP-BPL holds its position as the best scheme with respect to PU. Furthermore, we investigate the fundamental throughput-delay tradeoff in Fig. \ref{Fig7}. We plot the average packet delay for the PU versus its throughput at $P_{\mathrm{max}}/N_{0}=5$ dB. Intuitively, when a node needs to maintain a higher throughput, it loses in terms of the average delay encountered by its packets. Given that the system is stable, the node's throughput equals its packet arrival rate. Thus, increased throughput means injecting more packets into the system resulting in a higher delay. Furthermore, Fig. \ref{Fig7} shows that strictly lower average PU delay is attained via AP-based schemes compared to EP allocation in \cite{Krikidis}. It can also be noticed that AP-BPL is still in the leading position among all schemes in terms of both throughput and delay. Moreover, we validate the obtained closed-form expressions for average PU delay via simulations. Theoretical and simulation results for AP-BSL perfectly coincide. However, for AP-BPL, the slight deviation between theory and simulations is attributed to the relaxation of the constraint $\mathbf{h}_{\mathrm{I}} < \mathbf{h}_{r^*}$ discussed earlier.

Finally, we plot the average powers transmitted by the SUs in Fig. \ref{Fig8}, i.e., average $P_{s^*}$ and $P_{r^*}$, normalized to $N_{0}$, versus $P_{\mathrm{max}}/N_{0}$. Clearly, the AP-based schemes consume significantly less power than the EP assignment represented by the $45^ \circ$ line. For
the average power transmitted on the link $s^* \rightarrow D_s$, the first intuition that comes to mind is that AP-BSL policy results in the minimum average power. However, this is only true at high $P_{\mathrm{max}}/N_{0}$ values. It is noticed that the results corresponding to AP-BPL show
slightly less power consumption than that of AP-BSL at low $P_{\mathrm{max}}/N_{0}$ values. This behavior approximately holds till $P_{\mathrm{max}}/N_0=10$ dB. This is attributed to the nature of the proposed AP policy which sets $s^*$ silent if its maximum power constraint is not sufficient to satisfy the condition of success (\ref{P1}). Since in AP-BSL, $s^*$ always sees the best link to $D_{s}$, the number of slots in which it remains idle is less than that in AP-BPL. This yields a higher throughput at the expense of slightly higher average transmitted power. The same argument holds for comparing selection policies on the link $r^* \rightarrow D_p$.

\vspace{-3mm}
\subsection{Discussion on the Assumptions}\label{discussion}
The above system analysis is performed under the assumption of fully-backlogged SUs. The motivation behind this assumption is two-fold. First, backlogged SUs represent the worst case scenario from the PU's point of view. Since we consider cooperative communications, a portion of the
PU's data is delivered to its intended destination via the relay link, i.e., $r^* \rightarrow D_{p}$. However, the transmission of secondary
packets causes interference to the relay link as indicated earlier. This interference is persistent in case of backlogged SUs. Therefore, our results can be considered as a lower bound on the achievable performance of the PU, i.e., a lower bound on throughput and upper bound on delay. Furthermore, the
backlogged SUs assumption mitigates the interaction between the queues of the SUs. This renders the system mathematically tractable. Nevertheless,
stochastic arrivals to the SUs' queues can still be considered and queues interaction can be tackled using the dominant system approach originally introduced in \cite{rao}. However, this is out of the scope of the paper.

It is worth noting that in the derivations corresponding to BPL-based schemes, i.e., in Sections \ref{EP-BPL} and \ref{AP-BPL} of the Appendix, we consider $\mathbf{h}_{\mathrm{I}}$ and $\mathbf{h}_{r^*}$ independent random variables. However, they are coupled through the constraint $\mathbf{h}_{\mathrm{I}} < \mathbf{h}_{r^*}$. This constraint is an immediate consequence of the BPL node selection policy. We relax this constraint to render the problem mathematically tractable. Nevertheless, we quantify the effect of relaxing this constraint on the obtained closed-form expressions for $f_{r^*}$ through numerical simulation results presented in Section \ref{results}.

Finally, we assume that SUs perfectly sense the PU's activity. This assumption has been made to avoid adding further complexity to the analysis which might distort the main message behind the paper. Nevertheless, imperfect sensing has been studied extensively in the literature. Reference \cite{yucek2009survey} presents a comprehensive survey of spectrum sensing techniques in cognitive radio networks.

\begin{figure}[t]
	\centering
	\includegraphics[width=1\columnwidth , height=0.65\columnwidth]{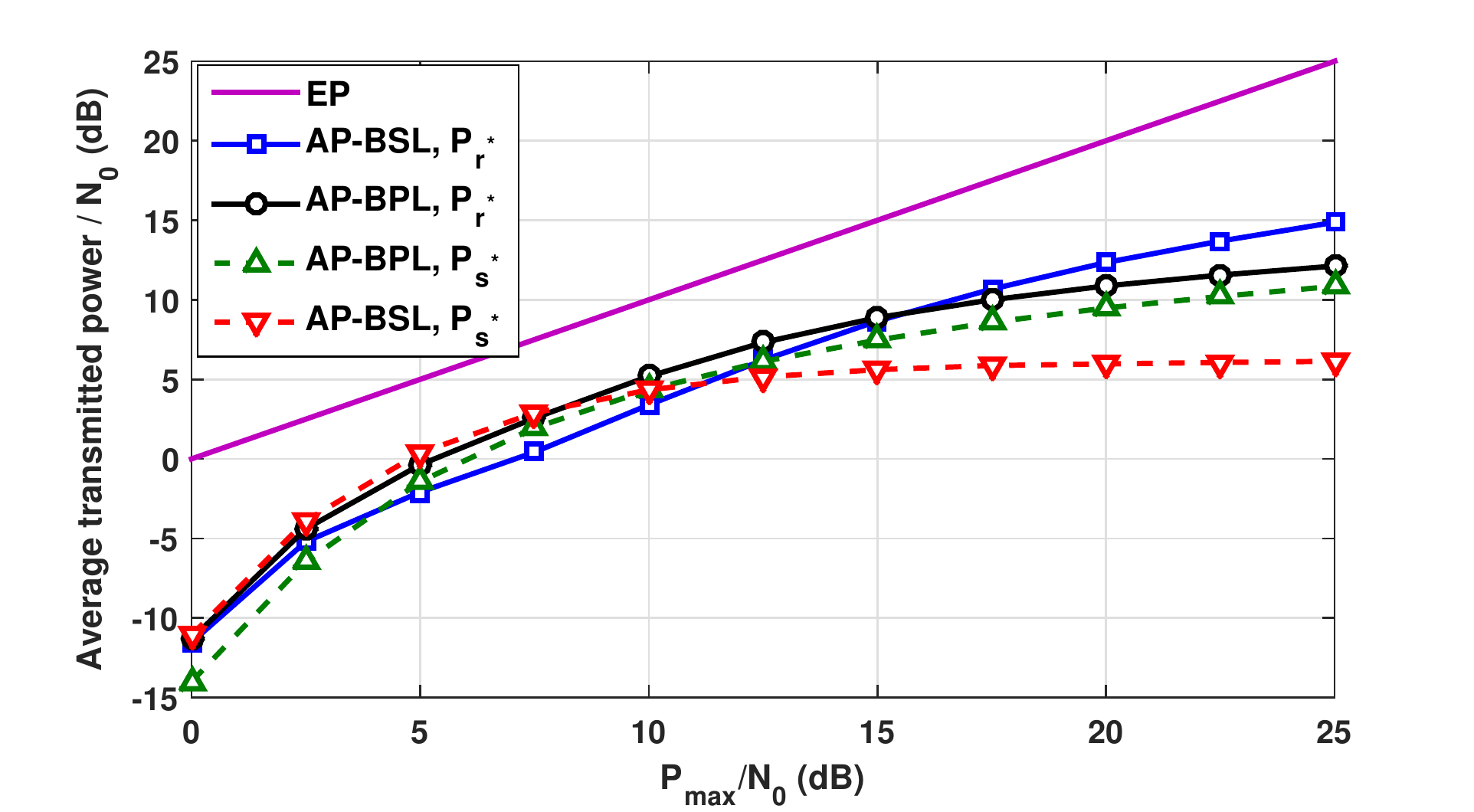}
	\caption{Average SUs' transmitted power normalized to $N_{0}$ versus $P_{\mathrm{max}}/N_0$.}
	\vspace{-3mm}
	\label{Fig8}
\end{figure}

\section{Conclusion} \label{sect:conclusion}
We discuss a power allocation policy for cognitive radio networks with multiple relays and propose different relaying protocols depending on the network utility function. The effect of SU power adaptation on
throughput and average delay is thoroughly investigated. We derive the closed-form expressions for the achieved throughput and average delay and validate the results through numerical simulations. Dynamically adapting the transmission powers at the SUs according to the channel conditions results in substantial improvement in primary and secondary throughput. The SUs under EP-based schemes always transmit at maximum power. This results in excessive interference on the relay link which is not the case for the AP-based schemes. Power adaptation is performed at the SUs to transmit with the minimum power required for the successful
transmission. To further benefit the system, the SUs back-off if their maximum permissible power is not sufficient to yield a successful transmission and avoid guaranteed outage events. The back-off benefits the other transmitting SU by reducing the incurred interference and thereby, causes throughput increase.
The AP-based schemes are shown to reduce the average queuing delay encountered by the PU packets compared to their EP-based counterparts.
We perform mathematical analysis of the proposed schemes and show numerically that the AP-based schemes save energy; and achieve higher throughput and lower delay simultaneously.

\appendices

\section{Distributions of $\mathbf{h}_{r^*}$, $\mathbf{h}_{\mathrm{I}}$, and $\mathbf{h}_{s^*}$ for BSL} \label{BSL_dist}
Referring to the policy described in Section \ref{BSL}, 
\begin{equation}
\mathbf{h}_{s^*}=\underset{i \in \{1,\hdots,N\}}{{\mathrm{max.}}}~\mathbf{h}_{s_{i}}.
\end{equation}
Therefore, the probability density function (PDF) of $\mathbf{h}_{s^*}$ is
\begin{equation}\label{Y_BSL}
\mathcal{P}_{\mathbf{h}_{s^*}}(h)=N e^{-h} (1-e^{-h})^{N-1},~h \geq 0.
\end{equation}     
As indicated earlier, the fact that $s^*$ has the best link to $D_{s}$ gives absolutely no information about its link quality to $D_{p}$ and hence, 
\begin{equation}\label{X_BSL}
\mathcal{P}_{\mathbf{h}_{\mathrm{I}}}(h)=e^{-h},~ h \geq 0.
\end{equation}
On the other hand, 
\begin{equation}\label{W_BSL}
\mathcal{P}_{\mathbf{h}_{r^*}}(h)=(N-1) e^{-h} (1-e^{-h})^{N-2},~ h \geq 0.
\end{equation} 
We present a rigorous argument to prove that (\ref{W_BSL}) is true. Consider the $2N$ random variables representing the link qualities of the $N$ SUs to $D_{p}$ and $D_{s}$. The SU with the best link to $D_{s}$ is selected to transmit a packet of its own. This leaves $(N-1)$ possible candidates for relaying a primary packet to $D_{p}$. Among the $(N-1)$ random variables representing the link qualities of these candidates to $D_{p}$, their maximum is selected. This maximum has one of the following two possibilites. 
\begin{itemize}
\item It is the second maximum of $\left\{\mathbf{h}_{r_{i}}\right\}_{i=1}^{N}$. This occurs only when the same SU has the best link to both $D_{p}$ and $D_{s}$ simultaneously. A specific SU has the best link to both destinations simultaneously with probability $1/N^2$. Taking into account $N$ such possibilities, one for every SU, $\mathbf{h}_{r^*}$ is the second maximum of $\left\{\mathbf{h}_{r_{i}}\right\}_{i=1}^{N}$ with probability $1/N$. 
\item It is the maximum of $\left\{\mathbf{h}_{r_{i}}\right\}_{i=1}^{N}$. This occurs whenever $s^*$ is not the SU having the best link to $D_{p}$, which has a probability $1-(1/N)$.       
\end{itemize}      
The average distribution corresponding to the two possibilities presented above with their respective probabilities is exactly the same as the distribution of a maximum of $(N-1)$ i.i.d. exponential random variables with means $1$ each. This is an easy-to-show fact using order statistics arguments, omitted for brevity. The proof of (\ref{W_BSL}) is then concluded.

\section{Distributions of $\mathbf{h}_{r^*}$, $\mathbf{h}_{\mathrm{I}}$, and $\mathbf{h}_{s^*}$ for BPL} \label{BPL_dist}
According to the policy described in Section \ref{BPL}, 
\begin{equation}
\mathbf{h}_{r^*}=\underset{i \in \{1,\hdots,N\}}{{\mathrm{max.}}}~\mathbf{h}_{r_{i}}.
\end{equation}
Therefore, the PDF of $\mathbf{h}_{r^*}$ is
\begin{equation}\label{W_BPL}
\mathcal{P}_{\mathbf{h}_{r^*}}(h)=N e^{-h} (1-e^{-h})^{N-1},~h \geq 0.
\end{equation}     
On the other hand, 
\begin{equation}\label{Y_BPL}
\mathcal{P}_{\mathbf{h}_{s^*}}(h)=(N-1) e^{-h} (1-e^{-h})^{N-2},~ h \geq 0.
\end{equation} 
An argument similar to that used to derive the distribution of $\mathbf{h}_{r^*}$ in Appendix \ref{BSL_dist} is used to derive (\ref{Y_BPL}).

The SU with the best link to $D_{p}$ is selected to relay a primary packet. This eliminates the possibility that $s^*$ has the best link to $D_{p}$, i.e., $\mathbf{h}_{\mathrm{I}}$ can not be the maximum of $\left\{\mathbf{h}_{r_{i}}\right\}_{i=1}^{N}$. In other words, $\mathbf{h}_{I}$ can possibly be the $k$th order statistic of the $N$ random variables $\left\{\mathbf{h}_{r_{i}}\right\}_{i=1}^{N}$, where $k=1,\hdots,N-1$. The $k$th order statistic is by convention the $k$th smallest random variable.   
It remains to note that after the selection of $r^*$, the remaining $(N-1)$ SUs possess equal probabilities of having the best link to $D_{s}$. Consequently, $\mathbf{h}_{\mathrm{I}}$ is equally likely to be any $k$th order statistic of $\left\{\mathbf{h}_{r_{i}}\right\}_{i=1}^{N}$, $k=1,\hdots,N-1$. Then, the average distribution of these order statistics is given by 
\begin{equation} \label{X_BPL}
\mathcal{P}_{\mathbf{h}_{\mathrm{I}}}(h)\!\!=\!\!\frac{N}{N-1}\!\! \displaystyle \sum_{k=1}^{N-1}
\!\!{\!N\!-\!1 \choose k\!-\!1\!}  e^{-h(N-k+1)}   (1\!-\!e^{-h})^{k-1},~h \geq 0. 
\end{equation}

\section{Derivation of $f_{r^*}$ and $f_{s^*}$ for EP-BSL} \label{EP-BSL}
Using (\ref{outage}) and (\ref{R2}) along with the description of power allocation and node selection policies provided in Sections \ref{UP} and \ref{BSL}, respectively, we have
\begin{equation}\label{f}
f_{r^*}=\mathbb{P}\left[ \mathbf{h}_{r^*} > a+\frac{\mathbf{h}_{\mathrm{I}}}{b} \right].
\end{equation}
Then, total probability theory implies that
\begin{align}\label{TPT}
f_{r^*} =\int_{0}^{\infty} \mathbb{P} \left[ \mathbf{h}_{r^*} > a+ \frac{h}{b} \right]
\mathcal{P}_{\mathbf{h}_{\mathrm{I}}}(h)dh 
\end{align}
Thus, (\ref{TPT}) is readily solved via substituting by the distributions of the random variables $\mathbf{h}_{\mathrm{I}}$ and $\mathbf{h}_{r^*}$ provided in (\ref{X_BSL}) and (\ref{W_BSL}), respectively. We first note that 
\begin{equation}\label{W_CDF_BSL}
\mathbb{P}\left[ \mathbf{h}_{r^*} > w \right]=1-(1-e^{-w})^{N-1}, ~ w \geq 0
\end{equation}
and then use (\ref{W_CDF_BSL}) with $w=a+ \frac{h}{b}$ in (\ref{TPT}) to get
\begin{align}\label{int_EP_BSL}
f_{r^*} = \int_{0}^{\infty} 
\left[ 1 - \left[ 1-e^{-\left(a+\frac{h}{b}\right)} \right]^{N-1} \right].
e^{-h} dh. 
\end{align} 
To solve this integration, we use the binomial theorem
\begin{equation} \label{binomial_EP_BSL}
\left[ 1-e^{-\left(a+\frac{h}{b}\right)} \right]^{N-1} = 
\displaystyle \sum_{k=0}^{N-1} {N-1 \choose k} (-1)^{k} e^{-k\left(a+\frac{h}{b}\right)}. 
\end{equation} 
We substitute by (\ref{binomial_EP_BSL}) in (\ref{int_EP_BSL}). Then, the integral solution renders $f_{r^*}$ as in (\ref{fkpDpBSL}).

At the SUs side, we depend on (\ref{outage}) and (\ref{R1}) to write
\begin{align}\label{fsksDs_BSL}
f_{s^*}= \mathbb{P}\left[ \mathbf{h}_{s^*} > a \right]
=1-\beta^{N}
\end{align}
which follows directly from (\ref{Y_BSL}). This verifies $f_{s^*}$ in (\ref{fksDsBSL}).

\section{Derivation of $f_{r^*}$ and $f_{s^*}$ for EP-BPL} \label{EP-BPL}
We use the description of power allocation and node selection policies presented in Sections \ref{UP} and \ref{BPL}, respectively. Using (\ref{outage}) and (\ref{R2}), $f_{r^*}$ is given by (\ref{f}) which is the same as (\ref{TPT}) through total probability theory. 
The distributions of $\mathbf{h}_{r^*}$ and $\mathbf{h}_{\mathrm{I}}$ given by (\ref{W_BPL}) and (\ref{X_BPL}), respectively, are used to solve the integral in (\ref{TPT}) using similar steps to that presented in Appendix \ref{EP-BSL}. This renders $f_{r^*}$ as given in (\ref{fskpDp_FPBPL}).

An SU transmits on the best link to $D_{s}$ only when $Q_{r}$ is empty. Therefore,
\begin{align}\label{2}
f_{s^*}&=\mathbb{P}\left[\left. \bar{\mathcal{O}}_{s^*} \right|\mathsf{B}\right]
\mathbb{P}\left[\mathsf{B}\right]
+\mathbb{P}\left[\left. \bar{\mathcal{O}}_{s^*} \right|\bar{\mathsf{B}}\right]
\mathbb{P}\left[\bar{\mathsf{B}}\right]
\end{align}
where $\mathcal{O}_{s^*}$ denotes the outage event on the secondary link, and $\mathsf{B}$ denotes the event that $Q_{r}$ is non-empty. A bar over an event's symbol denotes its complement. Little's theorem \cite{Bertsekas} implies that
\begin{equation}\label{busy}
\mathbb{P}\left[\mathsf{B}\right]=\gamma
\end{equation}
where $\gamma$ is given by (\ref{gamma}). In (\ref{busy}), we use the arrival and service rates of $Q_{r}$ presented on both sides of
(\ref{Qr}), respectively. Next, we compute the probability of packet success on the secondary link when $Q_{r}$ is busy. From (\ref{outage}) and (\ref{R1}), we have
\begin{equation}\label{f_Qr_busy}
\mathbb{P}\left[\left. \bar{\mathcal{O}}_{s^*} \right| \mathsf{B}\right]=
\mathbb{P}\left[\left. \mathbf{h}_{s^*} > a \right| \mathsf{B}\right]=1-\beta^{N-1}.
\end{equation}
This follows from the distribution of $\mathbf{h}_{s^*}$ given by (\ref{Y_BPL}). On the other hand, if $Q_{r}$ is empty, $s^*$ transmits on the best link among $\mathbb{S} \rightarrow D_s$,
i.e., $\mathbf{h}_{s^*} = \underset{i \in \{1,\hdots,N\}}{\mathrm{max.}}~h_{s_{i}}$. Thus, we have
\begin{equation}\label{f_Qr_empty}
\mathbb{P}\left[\left. \bar{\mathcal{O}}_{s^*}\right|\bar{\mathsf{B}}\right]=
\mathbb{P}\left[\left. \mathbf{h}_{s^*} > a \right| \bar{\mathsf{B}}\right]=1-\beta^N.
\end{equation}
We substitute by the results of (\ref{busy}), (\ref{f_Qr_busy}), and (\ref{f_Qr_empty}) in (\ref{2}). This verifies that $f_{s^*}$ is given by (\ref{fksDs_UPBPL}).

\section{Derivation of $f_{r^*}$ and $f_{s^*}$ for AP-BSL} \label{AP-BSL}
Using total probability theory, we write
\begin{align}\label{13}
f_{r^*}= \mathbb{P}\left[ \left. \bar{\mathcal{O}}_{r^*} \right|  \mathcal{O}_{s^*} \right]
\mathbb{P}\left[ \mathcal{O}_{s^*}\right] 
+ \mathbb{P}\left[ \left. \bar{\mathcal{O}}_{r^*} \right|  \bar{\mathcal{O}}_{s^*}\right]
\mathbb{P}\left[ \bar{\mathcal{O}}_{s^*}\right]
\end{align}
where $\mathcal{O}_{r^*}$ denotes the outage event on the relay link. 
In (\ref{13}), we take into account the fact that $s^*$ remains silent if $P_{\mathrm{max}}$ is not sufficient to satisfy (\ref{P1}). Therefore, we
compute the probability of a successful transmission on the relay link in both cases of $s^*$ activity, i.e., either active or silent. Thus, from (\ref{P1}), we have
\begin{align}\label{14}
\mathbb{P}\left[ \mathcal{O}_{s^*} \right]=
\mathbb{P}[\mathbf{h}_{s^*} < a ]=\beta^N.
\end{align}
This can directly be verified using the distribution of $\mathbf{h}_{s^*}$ presented in (\ref{Y_BSL}). 
In the event of a sure outage on the secondary link, $s^*$ refrains from transmission. 
We then plug $P_{s^*}=0$ into (\ref{P2}) and write
\begin{align}\label{15}
\mathbb{P}\left[ \left. \bar{\mathcal{O}}_{r^*} \right| \mathcal{O}_{s^*} \right]=
\mathbf{P}[\mathbf{h}_{r^*} > a ]=1-\beta^N.
\end{align}
This result is explained as follows. When $s^*$ is silent, $r^*$ is selected to be the SU with the best link to $D_{p}$ to enhance the PU throughput. Thus, in this specific case, $\mathbf{h}_{r^*}$ is
the maximum of $N$ exponential random variables with means 1 each. This renders $\mathbb{P}[\mathbf{h}_{r^*} > a]=1-\beta^N$.
 
On the other hand, when $s^*$ is active, i.e., $\mathbf{h}_{s^*} \geq a$, we choose $P_{s^*}$ to be the value that meets (\ref{P1}) with equality and plug it into (\ref{P2}). After some algebraic manipulation, we write
\begin{align}\label{16}
\mathbb{P}\left[ \left. \bar{\mathcal{O}}_{r^*} \right| \bar{\mathcal{O}}_{s^*}\right] \!=\!
\mathbb{P}\left[ \mathbf{h}_{\mathrm{I}} \leq \left. b\left( \frac{\mathbf{h}_{r^*}}{a}-1 \right)\mathbf{h}_{s^*} \right| \mathbf{h}_{s^*} \geq a \right]\!.
\end{align}
The first step towards solving (\ref{16}) requires the computation of $\mathbb{P}[\left. \mathbf{h}_{\mathrm{I}} \leq z\mathbf{h}_{s^*} \right| \mathbf{h}_{s^*} \geq a]$ for an arbitrary $z \geq 0$. Proceeding with that, we have
\begin{align}\label{102}
\mathbb{P}[\left. \mathbf{h}_{\mathrm{I}} \leq z\mathbf{h}_{s^*} \right| \mathbf{h}_{s^*} \geq a]=
\frac{\mathbb{P}[\mathbf{h}_{\mathrm{I}} \leq z \mathbf{h}_{s^*},\mathbf{h}_{s^*} \geq a]}
{\mathbb{P}[\mathbf{h}_{s^*} \geq a]}.
\end{align}
The numerator of (\ref{102}) can be computed as follows.
\begin{equation}\label{17}
\mathbb{P}[\mathbf{h}_{\mathrm{I}} \leq z \mathbf{h}_{s^*},\mathbf{h}_{s^*} \geq a]\!=\!\!
\int_{a}^{\infty} \!\!\!\! \int_{0}^{zy} \mathcal{P}_{\mathbf{h}_{\mathrm{I}}} (x)
\mathcal{P}_{\mathbf{h}_{s^*}} (y) dx dy
\end{equation}
The distributions of $\mathbf{h}_{\mathrm{I}}$ and $\mathbf{h}_{s^*}$ are given by (\ref{X_BSL}) and (\ref{Y_BSL}), respectively, and we use the fact that $\mathbf{h}_{\mathrm{I}}$ and $\mathbf{h}_{s^*}$ are independent. This information, along with the binomial theorem, is used to solve the double integral in (\ref{17}). Thus,
\begin{align}\label{18}
\mathbb{P}[\mathbf{h}_{\mathrm{I}} \leq z \mathbf{h}_{s^*},\mathbf{h}_{s^*} \geq a]=
N & \displaystyle \sum_{k=0}^{N-1} {N-1 \choose k} (-1)^{k} e^{-a(k+1)} \notag \\ 
&
\times
\left[ \frac{1}{k+1} - \frac{e^{-az}}{z+k+1} \right].
\end{align}
Furthermore, we know from (\ref{14}) that
\begin{equation}\label{101}
\mathbb{P}\left[ \bar{\mathcal{O}}_{s^*}\right]=
\mathbb{P}[\mathbf{h}_{s^*} \geq a]=1-\beta^N.
\end{equation}
Then, we substitute by (\ref{18}) and (\ref{101}) in (\ref{102}). Next,
we use total probability theory to write (\ref{16}) as
\begin{align}\label{19}
\mathbb{P}\!\left[\! \left. \bar{\mathcal{O}}_{r^*}\! \right| \bar{\mathcal{O}}_{s^*}\!\right]\!\!=\!\!\! 
\int_{a}^{\infty}\!\!\!\!\!
\mathbb{P}\!\left[ \mathbf{h}_{\mathrm{I}}  \leq\! \left. b\left( \frac{w}{a}-1 \right) \mathbf{h}_{s^*}
\right| \mathbf{h}_{s^*} \geq a \right]\!\mathcal{P}_{\mathbf{h}_{r^*}}\!(w)dw
\end{align}
where $\mathcal{P}_{\mathbf{h}_{r^*}}(.)$ is given by (\ref{W_BSL}). We then substitute by the result of (\ref{102}), with $z=b\left( \frac{w}{a}-1 \right)$, in (\ref{19}). The solution of the integral yields 
\begin{align}\label{20}
\mathbb{P}\!\left[\! \left. \bar{\mathcal{O}}_{r^*}\! \right|\! \bar{\mathcal{O}}_{s^*}\!\right]\!\!=\!\! 
\frac{N}{\left(1-\beta^N\right)}
\!\!\displaystyle \sum_{k=0}^{N-1}\!\! {\!N\!-\!1\! \choose\! k\!} (\!-1\!)^{k} e^{-a(k+1)} 
\!\left[ \mathrm{I}_{3}\! -\! \mathrm{I}_{4} \right]
\end{align}
where $\mathrm{I}_{3}$ and $\mathrm{I}_{4}$ are given by (\ref{I}) and (\ref{II}), respectively. The derivation of (\ref{II}) depends on the fact that
\begin{align}\label{integration}
\int_{a}^{\infty} \frac{e^{-tw}}{w+c}dw=e^{tc}E_{1}[t(a+c)]
\end{align}
for any constants $t$ and $c$. Substituting by (\ref{I}) and (\ref{II}) in (\ref{20}), and using (\ref{14}), (\ref{15}), (\ref{101}), and (\ref{20}) in (\ref{13}), $f_{r^*}$ is shown to be given by (\ref{fkpDpAPBSL}).

For the SUs, $f_{s^*}$ is shown to be given by (\ref{fsksDs_BSL}) following the same proof provided for the case of EP-BSL in Appendix \ref{EP-BSL}.

\section{Derivation of $f_{r^*}$ and $f_{s^*}$ for AP-BPL} \label{AP-BPL}
The derivation of $f_{r^*}$ for AP-BPL follows the same footsteps of the derivation presented in Appendix \ref{AP-BSL}. However, the difference in the node selection policies induces different distributions for the random variables of interest. We can write $f_{r^*}$ as in (\ref{13}). First, we derive the first term in the RHS of (\ref{13}) as follows.
\begin{align}\label{201}
\mathbb{P}\left[ \mathcal{O}_{s^*}\right]=
\mathbb{P}[\mathbf{h}_{s^*} < a]=\beta^{N-1}.
\end{align}
This follows from the distribution of $\mathbf{h}_{s^*}$ presented in (\ref{Y_BPL}). 
When $s^*$ is silent, we plug $P_{s^*}=0$ into (\ref{P2}) and write
\begin{align}\label{2000}
\mathbb{P}\left[\! \left. \bar{\mathcal{O}}_{r^*}\right| \mathcal{O}_{s^*}\right]=
\mathbb{P}[\mathbf{h}_{r^*}>a]=1-\beta^N
\end{align}
where the distribution of $\mathbf{h}_{r^*}$ is given by (\ref{W_BPL}).
Then, we shift our attention to the second term in the RHS of (\ref{13}).
When $s^*$ is active, i.e., $\mathbf{h}_{s^*} \geq a$, we choose $P_{s^*}$ to be the value that meets (\ref{P1}) with equality and plug it into (\ref{P2}).
Then, we compute the probability of success on the relay link given that $s^*$ is active as in (\ref{16}). 
We solve (\ref{17}) using the distributions of $\mathbf{h}_{s^*}$ and $\mathbf{h}_{\mathrm{I}}$ in (\ref{Y_BPL}) and (\ref{X_BPL}), respectively, along with the fact that they are independent to get
\begin{align}\label{202}
&\mathbb{P}\left[\mathbf{h}_{\mathrm{I}}\! \leq z \mathbf{h}_{s^*}\!,\mathbf{h}_{s^*}\! \geq a\right]\!\! = \!\!
\displaystyle \sum_{k=1}^{N-1} \sum_{\ell=0}^{k-1} \sum_{m=0}^{N-2} 
\!\!{\!N\! - \! 1 \!\choose k\! - \!1} {\!k\! - \! 1\! \choose \!\ell\!} {\!N\! - \! 2\! \choose\! m\!} \notag \\  
& \times \!\!\!
\frac{N(-1)^{m+\ell}}{(\! N \! - \! k \! + \! \ell \! + \! 1 \!)} \!\!
\left[\! 
\frac{e^{-a(m+1)}}{(m+1)}\! - \!
\frac{e^{-a\left(m+z(N-k+\ell+1)+1\right)}}
{\left( m \! + \!z  (  N \! - \! k \! +\! \ell \! + \!1 ) \! + \! 1\! \right)}\!
\right]  
\end{align}
for $z \geq 0$. Next, we substitute by the result of (\ref{202}), with $z=b\left(\frac{w}{a}-1\right)$, in (\ref{19}) and solve the integral. After some algebraic manipulation, omitted for brevity, the second term in the right hand side of (\ref{13}) is found to be equal to
\begin{align}\label{205}
\displaystyle \sum_{k=1}^{N-1} \sum_{\ell=0}^{k-1} \sum_{m=0}^{N-2}\!\!\!  
{\!N\!-\!1\! \choose\! k\!-\!1\!}\! {\!k\!-\!1\! \choose\! \ell\!} \!{\!N\!-\!2\! \choose\! m\!} 
\frac{(-\!1)^{m+\ell}N^2}{(N\!-\!k\!+\!\ell\!+\!1)} \left[ \mathrm{I}_{5}\! -\! \mathrm{I}_{6} \right]
\end{align}
where
\begin{align}
& \mathrm{I}_{5}=
\frac{e^{-a(m+1)}}{(m+1)} \int_{a}^{\infty} e^{-w}(1-e^{-w})^{N-1}dw  \label{I5_int}\\
& 
\mathrm{I}_{6}=\displaystyle \sum_{n=0}^{N-1} 
\frac{ a (-1)^{n} e^{-a(m+n+2-t)}}{(t-n-1)}
\int_{a}^{\infty} \frac{e^{-tw}}{w+c} dw \label{I6_int}
\end{align}
and the terms $t$ and $c$ are given by (\ref{t}) and (\ref{c}), respectively. The solution of the integral in (\ref{I5_int}) proves that $\mathrm{I}_{5}$ is given by (\ref{I5}). We use (\ref{integration}) to show that $\mathrm{I}_{6}$ is given by (\ref{I6}). Then, (\ref{201}), (\ref{2000}), and (\ref{205}) shows that $f_{r^*}$ is given by (\ref{fr_APBPL}).  

On the other hand, $f_{s^*}$ is shown to be given by (\ref{fksDs_UPBPL}) following the same proof provided for the case of EP-BPL in Appendix \ref{EP-BPL}.

\renewcommand{\baselinestretch}{1.5}

\bibliographystyle{IEEEtran}
\bibliography{IEEEabrv,bibliography}
\end{document}